\documentclass[journal]{new-aiaa}
\usepackage[utf8]{inputenc}
\usepackage{dsfont}
\usepackage[header,page,title]{appendix}
\usepackage[pagewise]{lineno}

\usepackage{graphicx}
\usepackage{float}

\usepackage{amsmath,bm,bbm,latexsym,subcaption,epstopdf}
\usepackage[version=4]{mhchem}
\usepackage[title]{appendix}
\usepackage{longtable,tabularx}
\usepackage[ruled,linesnumbered]{algorithm2e}
\usepackage[dvipsnames]{xcolor}
\usepackage{xcolor} 
\colorlet{blue}{black} 

\graphicspath{{./figures/}} 

\DeclareGraphicsExtensions{.pdf,.eps}

\newcommand{\dfn}{\triangleq}
\newcommand{\given}{\mid}

\DeclareMathOperator{\sign}{sgn}
\DeclareMathOperator{\sat}{sat}

\DeclareMathOperator{\expect}{\mathbb{E}}

\newcommand{\meas}{\ensuremath{\mathcal{Y}}}

\newcommand{\ZEM}{\ensuremath{Z_\text{DGL1}}}

\let\abs=\envert

\usepackage{amsthm}

\newtheorem{theorem}{Theorem}[section]

\newtheorem{lemma}[theorem]{Lemma}

\title{Unified Estimation–Guidance Framework \\ Based on Bayesian Decision Theory}
\author{Liraz Mudrik\footnote{This work is part of Dr. Mudrik’s Doctoral Research, Stephen B. Klein Faculty of Aerospace Engineering; currently Postdoctoral Fellow, Department of Mechanical and Aerospace Engineering, Naval Postgraduate School, Monterey, CA 93943; liraz109@gmail.com. Member AIAA (Corresponding Author).}}
\author{Yaakov Oshman\footnote{Professor Emeritus, Stephen B.\ Klein
    Faculty of Aerospace Engineering;
    \texttt{yaakov.oshman@technion.ac.il}. Fellow AIAA.
        \newline\newline Presented as Paper 2023-2497 at
    the AIAA SCITECH 2023 Forum, National Harbor, MD, January 23--27,
    2023; received 19 September 2025; accepted for publication 11
    February 2026; published online 10 April 2026.  Copyright
    \copyright\ 2026 by Liraz Mudrik and Yaakov Oshman.}}
\affil{Technion---Israel Institute of Technology, 3200003 Haifa, Israel\\
  {\upshape\url{https://doi.org/10.2514/1.G009628}}}

\begin{document}
\maketitle

\renewcommand{\include}{\input}

\begin{abstract}
  Using Bayesian decision theory, we modify the perfect-information,
  differential-game-based guidance law (DGL1) to address the
  inevitable estimation error occurring when driving this guidance law
  with a separately designed state estimator. This yields a stochastic
  guidance law complying with the generalized separation theorem, as
  opposed to the common approach, that implicitly, but unjustifiably,
  assumes the validity of the regular separation theorem. The required
  posterior probability density function of the game's state is
  derived from the available noisy measurements using an interacting
  multiple model particle filter. When the resulting optimal decision
  turns out to be nonunique, this feature is harnessed to
  appropriately shape the trajectory of the pursuer so as to enhance
  its estimator's performance. In addition, certain properties of the
  particle-based computation of the Bayesian cost are exploited to
  render the algorithm amenable to real-time implementation. The
  performance of the entire estimation--decision--guidance scheme is
  demonstrated using an extensive Monte Carlo simulation study.
\end{abstract}

\section{Introduction}
%





Using differential game theory to analyze the problem of guiding an
interceptor toward an evasively maneuvering target has been
intensively explored in recent
decades. Gutman~\cite{gutman_optimal_1979} developed a class of
differential-game-based guidance laws (DGLs) assuming that the
acceleration commands of both the interceptor and the target are
bounded, that the target has known linear time-invariant dynamics, and that both
players possess perfect information, i.e., full knowledge about the
state of the game.  
\textcolor{blue}{Specifically, a DGL law constitutes the saddle-point solution of a linearized zero-sum differential game of degree, with the terminal miss distance as the cost function.}
These DGLs rely solely on computing the time-to-go
(the time remaining until the end of the engagement) and the
instantaneous zero-effort miss (ZEM), which is the miss distance
resulting if both players do not apply any further acceleration
commands until the end of the game.  Assuming that the dynamics of
both players are strictly proper first-order, Shinar then developed the
so-called DGL1 law~\cite{shinar_solution_1981}. 
Assuming the perfect-information pattern, a DGL1-guided interceptor
that has superior maneuverability and agility capabilities \textcolor{blue}{(and, in special cases, even just superior agility)} is guaranteed to have hit-to-kill performance (i.e., deterministic
capture of the target). In this case, a singular region is formed in
the game space, whose boundaries depend on the time-to-go and the
(assumed known) maximal acceleration capabilities and time constants
of the dynamics of both players.  The DGL1 law exhibits a well-defined bang--bang behavior in the regular region (outside the singular
region).  Inside the singular region, the acceleration commands are arbitrary and are commonly chosen to be saturated linear so as to
prevent control chattering.

When the governing information pattern is not perfect, as is the case
in all realistic, stochastic scenarios, an estimator has to be relied
upon in order to estimate the state of the game. The inevitable
estimation error creates a crucial decision problem for an interceptor
implementing the DGL1 law (or similar perfect-information laws), as it
now has to decide on the state of the game and, consequently, on its
corresponding control action based on imprecise
information. Obviously, bad decisions have consequences, and we will
use the following two examples to explain what this statement means in
the context of a stochastic scenario where a perfect-information
guidance law is used.  Suppose that, due to the inherent estimation
error, the pursuer cannot determine which of the following two
hypotheses concerning the state of the game is true at a certain
point in time: 1) the state of the game is inside the singular region,
rendering any control action optimal, or 2) the state of the game is
in one of the regular regions, where the optimal control action is
unique.  
In this case, if the pursuer decides that the first
hypothesis is true when, in fact, it is not, it might apply a control action that will lead to \textcolor{blue}{an increased miss distance}, thus paying a dear price for its bad decision. 
On the other hand, deciding that the
second hypothesis is true will render the pursuer's control action
optimal---or, at least, nonconsequential---even if the first
hypothesis is true. Thus, in this case, the pursuer's decision-making
is easy.  Now consider another case in which the inherent uncertainty
leads to a decision between more hypotheses, involving at least two
different regular regions of the game space.  In this case, the
pursuer's decision-making becomes hard, as in each regular region the
optimal control action is unique, rendering any other control action
consequential and, potentially, disastrous in terms of miss
distance. Bad decisions, in this latter case, will have consequences.

Recognizing that the main problem of using a perfect-information law, such as DGL1, in a stochastic setting is its impaired ability to make optimal decisions in the presence of uncertainty, the new concept that we propose in this paper (a preliminary version of which has been presented in~\cite{mudrik_estimation-based_2019-1}) is to use Bayesian decision theory~\cite{trees_detection_2004} as a means of enhancing the capability of the guidance law to reach optimal decisions via statistically ranking the various guidance decisions possible at each point in time.  This process is based on the miss-distance-based costs and the a posteriori probabilities associated with each such decision.  The costs are computed using the DGL1 methodology, whereas the a posteriori probabilities are computed based on the posterior \textcolor{blue}{probability density function (PDF)} of the game's state, which, in our work, is generated by the interacting multiple model particle filter (IMMPF)~\cite{blom_exact_2007}. The choice of the IMMPF is based on its ability to cope with nonlinear and non-Gaussian models, as well as with non-Markovian mode-switching problems, which characterize realistic interception scenarios. Indeed, based on the findings of Shaferman and Oshman~\cite{shaferman_stochastic_2016}, according to which a well-timed evasion maneuver can enforce a considerable miss distance, a nonhomogeneous Markov model is assumed in this work for the target evasion maneuver mode to better model sophisticated targets.

When implemented naively, the new Bayesian decision-based guidance and estimation scheme can incur a severe computational burden.  This burden arises from the need to compute the Bayesian costs associated with the available decisions at each point in time and from the fact that, to compute these costs, the entire posterior PDF (represented, in our case, via a large particle population) is needed.  We significantly reduce the computational effort of the Bayesian decision criterion by observing that, throughout major parts of the interception engagement identified in this work, there is really no need to explicitly compute the costs in order to reach an optimal decision.

In some other parts of the engagement, the Bayesian decision criterion
may yield ambiguous decisions (recall that the perfect-information
DGL1 law also exhibits command ambiguity when the state of the game is
inside the singular region). Whereas, in the deterministic scenario,
this ambiguity merely allows using a linear command to prevent command
chattering, in our case we exploit it to find, in a fashion
reminiscent of~\cite{shaviv_estimation-guided_2017}, the optimal
command that would result in the best information-enhancing
trajectory, as this trajectory would yield improved information about
the state of the target, resulting in superior guidance performance.
The exploitation of decision ambiguity was first considered by the
authors in a preliminary conference
paper~\cite{mudrik_information-enhancement_2023}.

A noteworthy feature of the new guidance concept presented in this
work is that it complies with the guidelines of the generalized
separation theorem (GST), which Witsenhausen asserted
in~\cite{witsenhausen_separation_1971} based on a theorem by Striebel
in~\cite{striebel_sufficient_1965}. Valid for realistic scenarios
involving nonlinearities and non-Gaussian distributions, the GST
dictates that the guidance law should take into account the posterior
probability distribution of the state of the game, which, in turn,
should be generated via a separately designed estimator. Indeed, in
the new guidance and estimation scheme, the control commands are
generated by making optimal Bayesian decisions that rely on the entire
posterior PDF at each point in time. 
This is contrary to the common approach that straightforwardly lets a certainty-equivalent (perfect-information-based) guidance law operate on the output of a separately designed estimator, implying an assumption on the validity of the classical separation theorem~\cite{stengel_stochastic_1986}, which has never been proven valid for realistic scenarios.  Indeed, works such as \textcolor{blue}{~\cite{shinar_what_2003, shinar_integrated_2007, shaferman_stochastic_2016}} have shown that the ideal (perfect-information) performance
of the DGL1 law severely degrades under realistic (stochastic)
circumstances.  Guidance laws that utilize some posterior information
provided by a target state estimator have already been proposed, e.g.,
\cite{speyer_adaptive_1976,hexner_stochastic_2007,hexner_lqg_2008}. Later,
Shaviv and Oshman~\cite{shaviv_estimation-guided_2017} used the notion
of reachability sets to develop a GST-compliant methodology for the
stochastic interception problem, which is not limited by the standard
assumptions of linearity and Gaussian distributions.
\textcolor{blue}{The main contribution of this paper is the development of a unified, GST-compliant estimation and guidance framework that resolves the theoretical inconsistency inherent in assuming separation for this class of problems. To the best of the authors' knowledge, this work constitutes the first application of Bayesian decision theory to bridge the gap between differential-game-based guidance and stochastic estimation. This innovative approach not only prevents the performance degradation typical of certainty-equivalent implementations but also provides a foundation for estimation-aware guidance in non-Gaussian scenarios.
Consolidating the authors' preliminary results introduced in~\cite{mudrik_estimation-based_2019-1,mudrik_information-enhancement_2023}, the present work establishes the complete theoretical framework and derives the sufficiency conditions that reduce the computational burden to a level amenable to real-time implementation.}
Large-scale Monte Carlo (MC) simulations demonstrate that trajectory shaping can be systematically embedded within the full guidance law to yield consistent performance improvements. 

The remainder of this paper is organized as follows.  The following
section presents a mathematical formulation of the engagement
problem. The classical DGL1 guidance law is reviewed in
Sec.~\ref{sec:PI-DGL1}.  A brief description of the IMMPF algorithm is
presented in Sec.~\ref{sec:Estimation-process}, which also discusses
the use of nonhomogeneous Markov chains to model the mode-switching
dynamics.  Our main results are presented in
Sec.~\ref{sec:Est_awa_GL}, which applies Bayesian decision theory to
produce an estimation-aware version of the DGL1 law. This section also
shows how the computational effort associated with the new guidance
law can be significantly reduced, and how the decision ambiguity, when
it occurs, can be used to enhance information via trajectory shaping.
An extensive comparative MC simulation study is then presented in
Sec.~\ref{sec:sim}, examining the performance of the estimation-aware
and the classical DGL1 guidance laws in stochastic interception
scenarios.  Concluding remarks are offered in the final section.

\section{Problem Definition}
\label{Prob_Form}
A single-pursuer, single-evader interception scenario is
considered. Figure~\ref{fig:Planar-engagement-geometry} shows a
schematic view of the geometry of the assumed planar endgame scenario,
where $X_{I}$--$O_{I}$--$Y_{I}$ is a Cartesian inertial reference
frame. The interceptor and the target, as well as variables associated
with them, are denoted by $M$ and $T$, respectively.  The speed,
normal acceleration, and the path angle are denoted by $V$, $a$, and
$\gamma$, respectively.  The slant range between the interceptor and
the target is $\rho$, and the \textcolor{blue}{line-of-sight} (LOS) angle, measured
with respect to the $X_{I}$ axis, is $\lambda$.
\begin{figure}[tbh]
\begin{centering}
\includegraphics[width=3.25in]{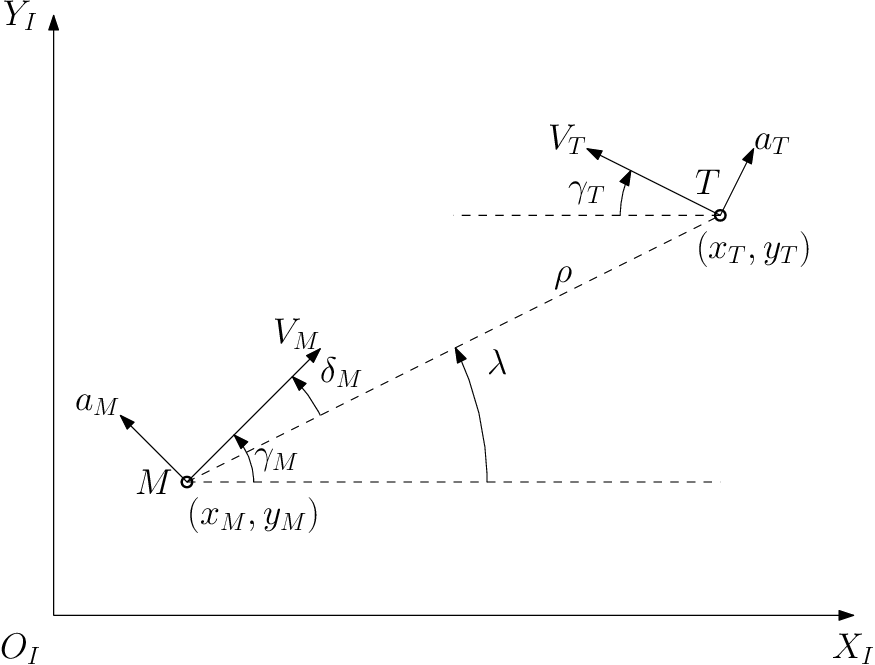}
\par\end{centering}
\caption{Planar engagement geometry.}
\label{fig:Planar-engagement-geometry}
\end{figure}

We use the following underlying assumptions:
\begin{enumerate}
\item Both the interceptor and the target are represented as point
  masses.
\item \label{ass2} The interceptor's own path angle and lateral
  acceleration are known (e.g., via the interceptor's own navigation
  system).
\item The speeds of the interceptor and the target, $V_{M}$ and $V_{T}$,
  respectively, are known and time-invariant.
\item The target's acceleration command, $u_{T}$, belongs to a known
  admissible set comprising $M$ possible maneuvers. For simplicity of
  exposition, we consider only the case of $M=2$, which means that the
  target performs bang--bang maneuvers (the generalization to $M > 2$
  is trivial). Known to be the optimal, perfect-information evasion
  strategy in bounded-acceleration pursuit--evasion games, its
  associated acceleration command is
\begin{equation}
	u_{T}=\begin{cases}
		+a_{T}^{max} & m=1\\
		-a_{T}^{max} & m=2
	\end{cases}
	\label{Switch_uT}
\end{equation}
where $m$ denotes the target maneuver mode.
\item \label{ass4} The interceptor and the target possess first-order
  dynamics with known time constants, $\tau_{M}$ and $\tau_{T}$,
  respectively.
\item \label{ass5} The lateral acceleration bounds of the interceptor
  and the target are known constants, $a_{M}^{\max}$ and $a_{T}^{\max}$,
  respectively.
\end{enumerate}

Following common practice, we use polar coordinates to present the
equations of motion (EOM). Defining the interceptor's state vector\textcolor{blue}{, used by its estimator,} as
\begin{equation}
\textbf{x}_{M} = \begin{bmatrix}
  \rho & \lambda & \gamma_{T} & a_{T}
\end{bmatrix}^{T}
\label{eq:StateVec1}
\end{equation}
yields the following EOM:
\begin{subequations}
\begin{align}
	\dot{\rho}       & =  V_{\rho}                                       \\
	\dot{\lambda}    & =  \frac{V_{\lambda}}{\rho}                       \\
	\dot{\gamma}_{T} & =  \frac{a_{T}}{V_{T}}                            \\
	\dot{a}_{T}      & =  -\frac{a_{T}}{\tau_{T}}+\frac{u_{T}}{\tau_{T}}
\end{align}
\label{eq:2}
\end{subequations}
where $u_{T}$ is the target's acceleration command, and
\begin{subequations}
\begin{align}
	V_{\rho}    & =  - (V_{M}\cos\delta_{M}+V_{T}\cos\delta_{T}) \\
	V_{\lambda} & = -V_{M}\sin\delta_{M}+V_{T}\sin\delta_{T}     \\
	\delta_{M}  & = \gamma_{M}-\lambda, \quad \delta_{T} = \gamma_{T}+\lambda
\end{align}
\label{eq:EOM1}
\end{subequations}

The interceptor's path angle and lateral acceleration satisfy the
following evolution equations:
\begin{subequations}
	\begin{align}
	\dot{\gamma}_{M} & =  \frac{a_{M}}{V_{M}}                            \\
	\dot{a}_{M}      & =  -\frac{a_{M}}{\tau_{M}}+\frac{u_{M}}{\tau_{M}}
	\end{align}
	\label{eq:EOM2}
\end{subequations}
where $u_{M}$ is the interceptor's acceleration command.

We assume that the interceptor can measure the bearing angle
$\delta_{M}$ between its own velocity vector and the LOS to the
target, rendering the following measurement equation:
\begin{equation}
y = \delta_{M}+\nu = \gamma_{M}-\lambda+\nu\label{eq:Meas}
\end{equation}
where $\nu$ is the (possibly non-Gaussian) measurement noise.

\section{Perfect-Information DGL1 Guidance Law}
\label{sec:PI-DGL1}
Derived within a differential game setting, the DGL1 guidance
law~\cite{gutman_optimal_1979,shinar_solution_1981} is the optimal
solution to the linearized case of perfect-information pursuit--evasion
games involving players with first-order dynamics and bounded
acceleration commands.  Because the DGL1 law constitutes the basis for
our new GST-compliant guidance law, we briefly review it in this
section.

The DGL1 guidance law is
\begin{equation}
u_{M}=  a_{M}^{\max}\sign(Z_{\text{DGL1}})
\label{Eq:OriginalDGL1}
\end{equation}
where the ZEM is computed according to~\cite{shaferman_stochastic_2016}
\begin{equation}
Z_{\text{DGL1}} (t) = \xi+\dot{\xi}t_{go}-a_{M}^{n}\tau_{M}^{2}\Psi
(\theta)+a_{T}^{n}\tau_{T}^{2}\Psi\big( \theta/\epsilon \big).
\label{eq:ZEM}
\end{equation}
In Eq.~\eqref{eq:ZEM}, $\xi$ and $\dot{\xi}$ are the relative displacement and relative velocity between the LOS and the target, measured along the normal to the instantaneous LOS, respectively.
\textcolor{blue}{Accordingly, the interceptor's and the target's normal accelerations, $a_{M}^{n}$ and $a_{T}^{n}$, are defined with respect to the instantaneous LOS angle $\lambda$:
\begin{subequations}
	\begin{align}
    a_{M}^{n} &= a_{M} \cos(\gamma_M - \lambda), \\
    a_{T}^{n} &= a_{T} \cos(\gamma_T + \lambda).
	\end{align}
\end{subequations}}
Additionally,
\begin{subequations}
	\begin{align}
	\Psi(\theta) &= e^{-\theta} + \theta - 1 \\
	\theta &= \frac{t_{go}}{\tau_{M}} \\
        \epsilon &= \frac{\tau_{T}}{\tau_{M}}
	\end{align}
\end{subequations}
The ZEM's evolution equation is
\begin{equation}
\dot{Z}_{\text{DGL1}}\left(t\right)=-\tau_{M}\Psi(\theta)u_{M}+\tau_{T}\Psi\big(\theta/\epsilon\big)u_{T}.
\label{eq:z_ODE}
\end{equation}

Let
\begin{equation}
\mu \dfn \frac{a_{M}^{\max}}{a_{T}^{\max}}
\end{equation}
be the maneuverability ratio.  Figure~\ref{fig:GameSpace} presents the
game space for the case where the interceptor possesses
maneuverability and agility superiority over the target, that is,
$\mu>1$ and $\mu\epsilon>1$.  The lines shown in this figure are
equicost lines, corresponding to cases where both the interceptor and
the target apply their respective optimal strategies until the end of
the game, thus enforcing a time-invariant ZEM, which, at the end of
the game, becomes the final miss distance.  Corresponding to zero miss
distance, the two solid lines define the boundary of the singular
region, denoted as $\mathcal{Z}^{*}$. When the game's state is inside
$\mathcal{Z}^{*}$, \emph{any} pursuer acceleration command is
optimal in the sense that it guarantees hit-to-kill performance.
Because of this nonuniqueness, a linear command is commonly used to
prevent control chattering, rendering the following practical form of
the DGL1 law:
\begin{equation}
u_{M} = \begin{cases}
a_{M}^{\max} \sat\big(\frac{Z_{\text{DGL1}}}{k\partial\mathcal{Z}^{*}}\big) & Z_{\text{DGL1}}\in\mathcal{Z}^{*}\\
a_{M}^{\max} \sign( Z_{\text{DGL1}}) & Z_{\text{DGL1}} \notin \mathcal{Z}^{*}\\
\end{cases}
\label{Eq:DGL1}
\end{equation}
where $\sat(\cdot)$ stands for the saturation function, $0 < k \leq 1$
is the portion of the singular region in which the acceleration is
(arbitrarily set to be) linear, $\partial\mathcal{Z}^{*}$ is the
boundary of the singular region, computed as
\begin{equation}
	\partial \mathcal{Z}^{*}  = a_{M}^{\max} \tau_{M}^{2} \Upsilon(\theta)
	- a_{T}^{\max} \tau_{T}^{2} \Upsilon(\theta / \epsilon)
	\label{eq:sing_reg}
      \end{equation}
      and
      \begin{equation}
	\Upsilon(\theta) = \frac{1}{2} \theta^{2} - e^{-\theta} - \theta + 1.
\end{equation}
\begin{figure}[tb]
	\includegraphics[width=3.25in]{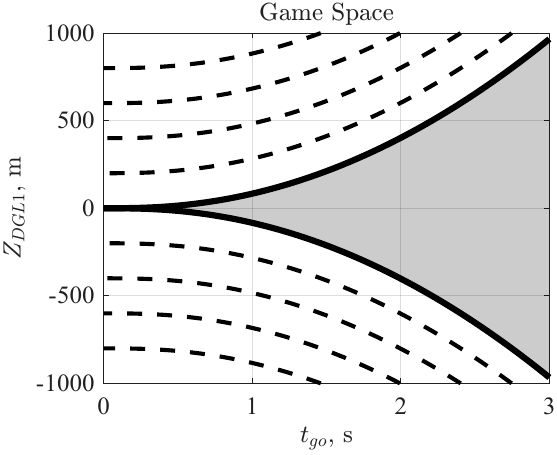}
	\centering
	\caption{DGL1 game space \textcolor{blue}{and its singular region (gray)}. Equicost level lines (dashed) and boundary lines (corresponding to zero miss distance) of the singular region, $\partial \mathcal{Z}^{*}$ (solid).}
	\label{fig:GameSpace}
\end{figure}

The calculation of the ZEM requires the time-to-go, which cannot be
computed accurately in real time. Therefore, the following
approximation is used:
\begin{equation}
t_{go}\approx-\frac{\rho}{V_{\rho}}\label{eq:t_go}
\end{equation}
along with the following relation, used in a nonlinear setting when
linearizing about the instantaneous LOS:
\begin{equation}
\xi = \rho \sin \lambda.
\end{equation}
Thus, in a nonlinear setting, the ZEM of the well-known proportional
navigation (PN) guidance law is approximated by
\begin{equation}
Z_{\text{PN}}=\xi+\dot{\xi}t_{go}\approx
-V_{\rho}t_{go}^{2}\dot{\lambda}\cos\lambda
=t_{go}V_{\lambda}\cos\lambda
\label{eq:Z_PN}
\end{equation}
and the ZEM of the DGL1 law is computed via
\begin{equation}
Z_{\text{DGL1}} = Z_{\text{PN}}+Z_{\text{acc}}
\label{eq:Z_formula}
\end{equation}
where
\begin{equation}
Z_{\text{acc}} \dfn -a_{M}^{n}\tau_{M}^{2}\Psi(\theta)+a_{T}^{n}\tau_{T}^{2}\Psi\big(\theta/\epsilon\big).
\label{eq:Z_acc}
\end{equation}

The DGL1 game space is constructed solely from the ZEM and the
time-to-go, as shown in Fig.~\ref{fig:GameSpace}. To directly estimate the ZEM and the time-to-go, we define a new state vector \textcolor{blue}{denoted by $\bm{\chi}_M$, which differs from the physical state vector $\bm{x}_M$ defined in Eq.~\eqref{eq:StateVec1}}:
\begin{equation}
\textcolor{blue}{\bm{\chi}_{M}} = \begin{bmatrix}
t_{go} & Z_{\text{DGL1}} & \lambda & \gamma_{T}
\end{bmatrix}^{T}.
\label{StateVec}
\end{equation}
This yields the EOM as
\begin{subequations}
	\label{eq:EOM_new}
	\begin{align}
	\dot{t}_{go} & = -1-\frac{\left(a_{M}\sin\delta_{M}+a_{T}\sin\delta_{T}\right)t_{go}}{V_{\rho}}+\left(\frac{V_{\lambda}}{V_{\rho}}\right)^{2}\\
	\dot{Z}_{\text{DGL1}}& =
	\dot{Z}_{\text{PN}} + \dot{Z}_{\text{acc}} \\
	\dot{\lambda} & = \frac{V_{\lambda}}{\rho}\\
	\dot{\gamma}_{T} & = \frac{a_{T}}{V_{T}}
	\end{align}
\end{subequations}
\noindent where $\dot{Z}_{\text{PN}}$ and $\dot{Z}_{\text{acc}}$ are computed by
differentiating Eqs.~\eqref{eq:Z_PN} and \eqref{eq:Z_acc} with respect
to time, respectively. Alternatively, Eq.~\eqref{eq:z_ODE} can be used
as the EOM of $\dot{Z}_{\text{DGL1}}$ in the linearized
case. Moreover, using Eqs.~\eqref{eq:t_go} and
\eqref{eq:Z_formula}, we can compute $\rho$ and $a_{T}$ via
\begin{subequations}
	\begin{align}
	\rho= & -V_{\rho}t_{go} \\
	a_{T}= & \frac{Z_{\text{DGL1}}-t_{go}V_{\lambda}\cos\lambda+a_{M}^{n}\tau_{M}^{2}\Psi\left(\theta\right)}{\cos\delta_{T}\tau_{T}^{2}\Psi\big(\theta/\epsilon\big)}.
	\end{align}
	\label{eq:trans}
\end{subequations}

\section{Target-Tracking IMMPF}
\label{sec:Estimation-process}
In this section, we present the estimation algorithm that we use to drive
the new stochastic guidance law to be derived in the sequel.  We first
present a generic IMMPF algorithm. Then, we adapt this algorithm to
suit the problem at hand by addressing time-varying transition
probabilities and a particular initialization of the filter.
\subsection{IMMPF Algorithm}
The IMMPF algorithm is a multiple-model sequential MC
estimator~\cite{blom_exact_2007}.  We select the IMMPF for its ability
to cope with nonlinear dynamics, non-Gaussian driving noises, and
multiple models with non-Markovian switching modes.  
\textcolor{blue}{Crucially, unlike the Kalman filter and its variants, the IMMPF provides a particle-based estimate of the entire posterior PDF at each time step. This feature is essential because we explicitly exploit the full probability distribution in the guidance law design.}
Generally following the
logic of the IMM estimator~\cite{blom_interacting_1988}, in the IMMPF,
the resampling and interaction stages are combined before the
prediction and filtering stages.  The filter runs a bank of particle
filters (PFs) matched to all possible modes.  At each time $t_{k}$, the
PDF is represented by the set of particles and associated \textcolor{blue}{scalar} weights
$\{ \textbf{x}_{k}^{m,s}, w_{k}^{m,s} \mid m\in\mathbb{M}, s\in \{1,
\dots, S\}\} $, where $\textbf{x}_{k}^{m,s}$ is the state vector of
particle $s$ and mode $m$, and $w_{k}^{m,s}$ is its associated weight.
$\mathbb{M}$ is the set of $M$ discrete modes, and $S$ is the number
of particles of each mode, rendering the total number of particles
$N_{p}=MS$.
\subsection{Target-Tracking IMMPF}
\label{subsec:TTIMMPF}
To implement the IMMPF in our problem, we use the models presented
in Sec.~\ref{Prob_Form}.  
Two discrete modes are used for the target
acceleration command in Eq.~\eqref{Switch_uT}.  The prediction step
uses the nonlinear EOM in Eq.~\eqref{eq:EOM_new}, recast in
discrete time as
\begin{equation}
\textbf{x}_{k} = \textbf{f}_{k-1} (\textbf{x}_{k-1},m_{k},\textbf{w}_{k})
\label{eq:propa}
\end{equation}
where $\textbf{x}_{k}$ is the interceptor relative state at time
$t_{k}$ and $m_{k}$ is the target mode during the time interval
$( t_{k-1}, t_{k}]$. The covariance of the zero-mean process noise
$\textbf{w}_{k}$ is used as a tuning parameter of the filter.  The
filter step uses the measurement equation~\eqref{eq:Meas} to calculate
the likelihood density $p(\textbf{z}_{k} \given \textbf{x}_{k},m_{k})$.
\textcolor{blue}{Note that the particle propagation step uses the nonlinear dynamics derived in Eq.~\eqref{eq:EOM_new}. Since the interceptor's own acceleration $u_M$ is known, it is treated as a deterministic input in the filter process model, ensuring that the state evolution is consistent with the engagement kinematics.}
\subsubsection{Mode Estimation}
\label{subsec:ModeEst}
The target acceleration command sequence is commonly assumed to obey a
homogeneous Markov switching model.  This assumption presupposes that
the target's flight controller is indifferent with regard to the
timing of the acceleration switch maneuver, which is equivalent to
assuming that the target cannot affect the resulting miss distance by
properly timing this maneuver. This assumption, however, has been
shown to be wrong in, e.g.,~\cite{shaferman_stochastic_2016}, which
demonstrates that a ``smart'' target can enforce a large miss distance
by properly timing its acceleration switch maneuver. \textcolor{blue}{To} address such a target, we use a nonhomogeneous Markov chain model for the \textcolor{blue}{transition probability matrix (TPM):}
\begin{equation}
  \Pi(t_{k})=
  \begin{bmatrix}
  	1-\pi^{12}_{k} & \pi^{12}_{k}   \\
  	\pi^{21}_{k}   & 1-\pi^{21}_{k}
  \end{bmatrix}
  \label{eq:TVTPM}
\end{equation}
where \textcolor{blue}{
\begin{equation}
    \pi_{k}^{ij} \triangleq \Pr(m_k=j | m_{k-1}=i), \quad \forall i,j=1,2
\end{equation}
are the transition probabilities} at time $t_{k}$.


The time-varying TPM allows us to better cope with sophisticated
targets capable of using optimal evasive maneuvers, e.g., targets
using the guidance law presented
in~\cite{shaferman_near-optimal_2021}, which exploits the
interceptor's inherent estimation error. By increasing the transition
probabilities over the optimal evasive maneuver's time window, the
effect of this maneuver can be mitigated.
\subsubsection{Filter Initialization}
We assume that the interceptor's estimator is initialized by a
radar. Figure~\ref{fig:init-radar-geometry} presents a schematic view of
the assumed planar geometry of the initializing radar, the interceptor,
and the target. Variables associated with the radar are denoted by the
additional subscript R. The range between the radar and the target is
denoted by $\rho_{R}$, and $\lambda_{R}$ is the angle between the
radar's LOS to the target and the $X_{I}$ axis.
\begin{figure}[tb]
	\begin{centering}
		\includegraphics[width=3.25in]{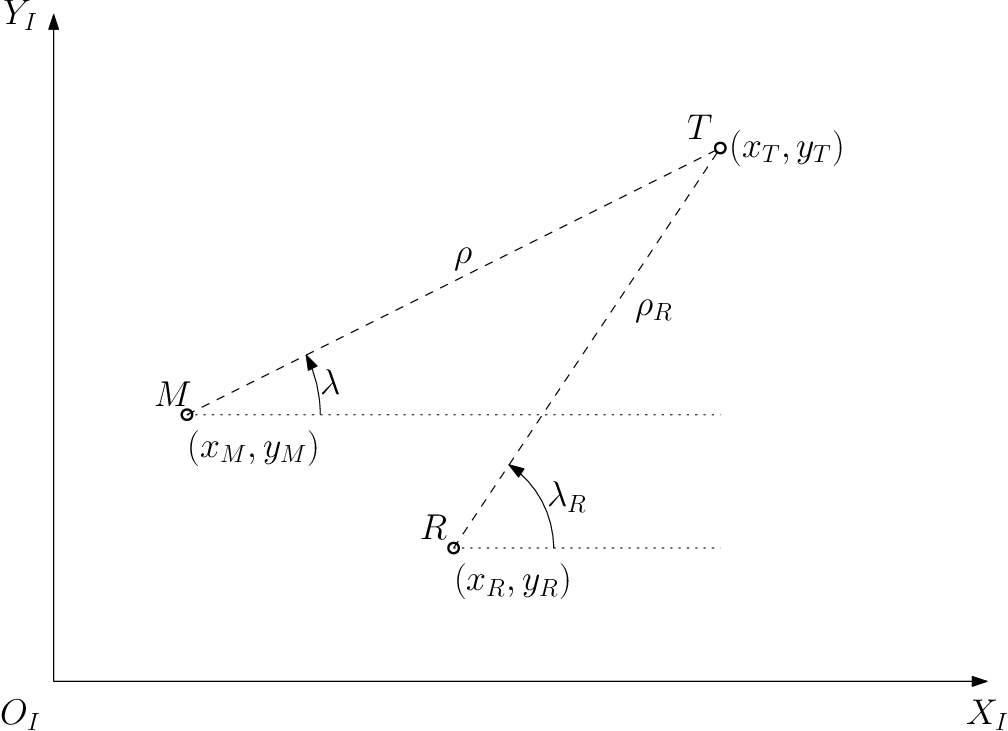}
		\par\end{centering}
	\caption{Initializing radar geometry.}
	\label{fig:init-radar-geometry}
\end{figure}

The radar measures the following vector:
\begin{equation}
  \textbf{x}_{R} = \begin{bmatrix}
    \rho_{R} & \lambda_{R} & \gamma_{T} & a_{T}
  \end{bmatrix}^{T}
\label{eq:RadarSV}
\end{equation}
This vector is then passed on to the interceptor's estimator. We
further assume that the radar's position $\left(x_{R},y_{R}\right)$
and the interceptor's position at initialization
$\left(x_{M},y_{M}\right)$ are known. 
The geometric relation between the interceptor's relative state and the
radar measurement is
\begin{subequations}
	\begin{align}
          \rho &=  \sqrt{ \rho_{R}^{2} + \Delta R_{R}^{2} + 2\rho_{R} [\Delta X_{R} \cos(\lambda_{R}) + \Delta Y_{R} \sin(\lambda_{R})]}\\
          \lambda &= \arctan \left[ \frac{\Delta Y_{R} + \rho \sin(\lambda_{R})}{\Delta X_{R} + \rho \cos(\lambda_{R})} \right]
	\end{align}
	\label{eq:RadarInit}
\end{subequations}
where
\begin{equation}
\Delta X_{R} \dfn x_{R} - x_{M}, \quad 
\Delta Y_{R} \dfn y_{R} - y_{M}, \quad 
\Delta R_{R} \dfn \sqrt{ \Delta X_{R}^2 + \Delta Y_{R}^2 }.
\end{equation}
Using Eqs.~\eqref{eq:t_go} and \eqref{eq:Z_formula} then enables
initialization of the state vector in Eq.~\eqref{StateVec}.
\textcolor{blue}{These geometric relations enable the direct
  initialization of the particle cloud from the radar
  estimate. Specifically, we draw random samples from the radar
  estimate's PDF and map them through Eqs.~\eqref{eq:RadarInit} to
  generate the initial state particles for the filter.}

\section{Estimation-Aware Guidance}
\label{sec:Est_awa_GL}
In this section, we derive this paper's main result: a GST-compliant, computationally efficient stochastic guidance law that is based on the perfect-information DGL1 guidance law. 
\textcolor{blue}{We begin by presenting Bayesian decision theory, inspired by its derivation in~\cite{trees_detection_2004}, and then frame the stochastic guidance problem as a multihypothesis Bayesian decision problem.}
We then proceed by deriving conditions for making optimal Bayesian decisions at a minimal computational cost and conclude this section by augmenting the new guidance law with an information-enhancing trajectory-shaping algorithm that exploits the decision ambiguity frequently occurring in the early stages of the interception.
\textcolor{blue}{
\subsection{Bayesian Decision Theory}
The decision problem seeks to find an optimal rule that chooses the
optimal hypothesis by minimizing a risk function over $n$ given
hypotheses $\{ H_i \}_{i=1}^n$.
Following~\cite{trees_detection_2004}, we define the Bayesian risk
$\mathcal{R}$ to be the expected value of the cost function.  Given
the set of available measurements $\meas$, we seek to minimize the
conditional risk (or expected posterior loss):
\begin{equation}\label{eq:riskdef}
  \mathcal{R}(H_i \given \meas) \dfn \expect\left[J(H_i) \given \meas\right]
  =\sum_{j=1}^{n} C_{ij} \Pr(H_{j} \given \meas),
\end{equation}
where $J(H_i)$ is the cost resulting from deciding $H_i$; $C_{ij}$ is the cost incurred when the decision rule decides
that hypothesis $H_{i}$ is true when, in fact, hypothesis $H_{j}$ is
true; and $\Pr(H_{j} \given \meas)$ is the posterior probability that
hypothesis $H_j$ is true.}

\textcolor{blue}{
To facilitate the minimization, we decompose the total risk into a fixed, unavoidable component (the risk incurred even if the true hypothesis were known) and a variable, decision-dependent component. We define the expected additional risk $\mathcal{B}_i$ as the difference between the (decision-dependent) conditional risk of deciding $H_i$ and the (unavoidable) optimal conditional risk, obtained as the sum of all possible correct decisions~\cite{trees_detection_2004}:
\begin{equation}
  \mathcal{B}_i(\meas) \dfn \mathcal{R}(H_i \given \meas) - \sum_{j=1}^{n} C_{jj} \Pr(H_{j} \given \meas).
\end{equation}
Substituting Eq.~\eqref{eq:riskdef} into this definition yields
\begin{equation}\label{eq:add_risk}
    \mathcal{B}_i(\meas) = \sum_{j=1}^{n} \Pr(H_{j} \given \meas) (C_{ij} - C_{jj}),
\end{equation}
where it is noted that the summand corresponding to $j=i$ vanishes.}

\textcolor{blue}{ Now, using Bayes's rule, the posterior probability
  can be decomposed to explicitly exhibit the contribution of the
  prior knowledge and the contribution of the evidence provided by the
  measurements. Thus,
\begin{equation}\label{eq:bayes_decomp}
  \Pr(H_{j} \given \meas) \propto P(\meas \given H_j) P_j,
\end{equation}
where $P_j$ denotes the prior probability of hypothesis $H_j$, and
  $P(\meas \given H_j)$ denotes the likelihood of observing the
  measurements conditional on hypothesis $H_j$ being true.  Using the
posterior decomposition~\eqref{eq:bayes_decomp} in
Eq.~\eqref{eq:add_risk} and defining the unnormalized additional
  risk function $I_i(\meas)$ to be
\begin{equation}
  I_{i} (\meas) \dfn \sum_{\substack{j=1 \\ j\neq i}}^{n}
  P_{j} P(\meas \given H_j) (C_{ij} - C_{jj}), \quad i=1,\dots,n.
    \label{Eq:I}
\end{equation}
renders the optimal Bayesian decision rule as
  \begin{equation}
 \text{ Decide } H_{i^\star},  \text{ where }   i^{\star}
         \dfn \arg\min_{i\in \{1,...,n\} }\{ I_{i}(\meas)\}.
    \label{Eq:Decision}
  \end{equation}} \textcolor{blue}{
\subsection{Stochastic Guidance as a Bayesian Decision Problem}
\label{subsec:Bayes4int}
As Eq.~\eqref{Eq:DGL1} shows, the DGL1 guidance law depends on
knowing, at each point in time, where the state of the game is: inside
the singular region, above it, or below it (see
Fig.~\ref{fig:GameSpace}).  The commands computed for these states are
vastly different from each other, and, clearly, computing a guidance
command appropriate for one state of the game (e.g., inside the
singular region) when, in fact, another state (e.g., outside and above
the singular region) is valid might lead to catastrophic
results, i.e., a significant increase in the miss distance. Under the perfect-information assumption, identifying the
true state of the game is not an issue, and the perfect-information
guidance law can be precisely computed and implemented.  However, in
the stochastic case, an estimator has to be implemented, and the state
of the game is not deterministically known but instead has to be
estimated. In this case, the determination of the state of the game
becomes a hard problem, admitting only a probabilistic answer.}

\textcolor{blue}{ 
To address the problem in a concrete way, consider again Fig.~\ref{fig:GameSpace} and define the following four hypotheses about the state of the game. 
\begin{description}
\item[$H_{1}$:] the state is above the singular region.
\item[$H_{2}$:] the state is inside the singular region, and the
  target acceleration mode is $m = 1$.
\item[$H_{3}$:] the state is inside the singular region, and the
  target acceleration mode is $m = 2$.
\item[$H_{4}$:] the state is below the singular region.
\end{description}
In the regular regions, $H_1$ and $H_4$, the target's optimal strategy is uniquely determined by the game solution~\eqref{Eq:DGL1}. In contrast, within the singular region, the target's optimal control is nonunique; thus, $H_2$ and $H_3$ are distinguished explicitly by the target's maneuver mode.
To illustrate these hypotheses, as well as the utility and function of
the IMMPF filter, Fig.~\ref{fig:part_Hyp_map} superimposes the
particle cloud that the IMMPF maintains to represent the posterior
PDF at a particular moment in time on a part of the DGL1 game
space. The particles are color-coded according to the hypotheses with
which they are associated.
\begin{figure}[tbh]
\begin{center}
\includegraphics[width=3.25in]{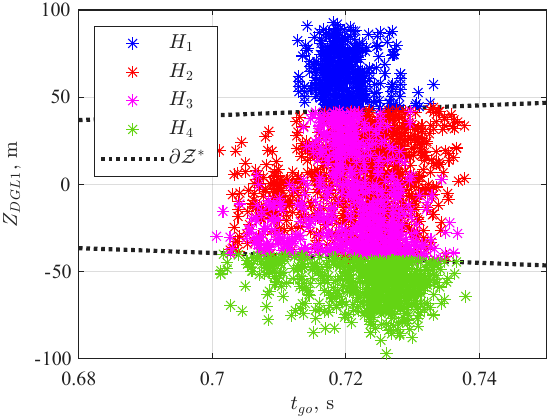}
\end{center}
\caption{\textcolor{blue}{IMMPF particle cloud distribution, color-coded by hypothesis, relative to the singular region boundaries at a particular moment during the interception.}}
\label{fig:part_Hyp_map}
\end{figure}}

\textcolor{blue}{
As can be easily seen from Fig.~\ref{fig:part_Hyp_map}, in the
stochastic case there is no deterministic answer to the following question:
Where is the state of the game at a particular moment in time?
Instead, we can talk about the probability that the state resides in a
specified zone in the game space. Thus, straightforwardly implementing
a perfect-information guidance law, such as the DGL1 law, is both
conceptually and practically wrong. Because the unavoidable
uncertainty regarding the state of the game can lead to catastrophic
results, we employ Bayesian decision theory, which explicitly takes into
account the uncertainty and the costs of wrong decisions. }

\textcolor{blue}{
To frame the problem as a Bayesian decision one, we define the cost
function $J$ to be the final miss distance of the stochastic
interception problem.
\textcolor{blue}{In the context of a zero-sum differential game with miss distance as the payoff, the cost function represents the guaranteed miss distance. When the state lies in the regular region (outside the singular region), the value of the game is exactly the distance from the current state to the boundary of the singular region.} 
Therefore, we define the cost $C_{ij}$ to be the final miss distance obtained when the guidance law assumes, for some finite horizon, that hypothesis $H_{i}$ is true when, in fact, hypothesis $H_{j}$ is true.  Based on its definition, the cost is computed based on the location of the state of the game within the game space.  Thus, in our specific case, when the state is such that the associated ZEM is inside the singular region, the cost is zero. When the ZEM is outside the singular region, the cost is the distance between the current ZEM and the singular region's nearest boundary.}

\textcolor{blue}{
Applying the Bayesian decision rule~\eqref{Eq:Decision} to the
interception problem involves the computation of the
likelihoods, the decision costs, and the prior
probabilities of the hypotheses. We do these calculations using the
IMMPF particle states and weights, as detailed in the following.}

\begin{description}
\item[\textcolor{blue}{Likelihood Functions.}]
  \textcolor{blue}{Having the particle-based density approximation
      generated by the IMMPF on hand, we approximate the likelihoods
    by summing the IMMPF weights of the particles associated with
    each hypothesis.  We denote the weight of the $j'$-th particle at
    time $k$ by $w_k^{j'}$.  Consequently,
    \begin{equation}
    \label{eq:likelihood}
    P(\meas \given H_j) = \sum_{j'\in H_{j}} {w}^{j'}_{k}, \quad j \in \{1, \dots,n\}
    \end{equation}
    To gain some insight into the effect of this term, consider the
    (extreme) case where the particles are entirely concentrated in
    $H_\ell$ for some $\ell \in \{1,\dots,n\}$.  In this case,
    Eq.~\eqref{eq:likelihood} implies
  \begin{equation}
    \label{eq:likelihood_example}
    P(\meas \given H_j) =
    \begin{cases}
      1 & j=\ell  \\
      0 &  j\ne \ell 
      \end{cases}
  \end{equation}
  and Eq.~\eqref{Eq:I} yields 
  \begin{equation}
    \label{eq:I_example}
    I_i =
    \begin{cases}
      P_\ell (C_{i\ell }-C_{\ell \ell }) & i\ne \ell  \\
      0  &  i=\ell 
    \end{cases}
  \end{equation}
  Noting that $C_{i\ell }-C_{\ell \ell }\ge 0$, by our earlier assumption, then
  leads to $i^\star = \ell $, and the optimal decision is $H_\ell $.}

\item[\textcolor{blue}{Costs.}]
  \textcolor{blue}{In calculating the costs, we have to separately address the following four distinct cases:}
	\begin{enumerate}
        \item \textcolor{blue}{The cost of correctly choosing hypothesis $H_{j}$ when
          $P(\meas \given H_j)>0$
          is 
		\begin{equation}
			C_{jj} =  \sum_{j'\in H_{j}} \tilde{w}^{j'}_{k}
			\max (0, \abs{ z_{j'}(t_{go;j'})} - z^{*}(t_{go;j'}))
			\label{eq:c_jj}
		\end{equation}
                where
                $z_{j'}$ is the ZEM corresponding to particle
                $j'$ that belongs to hypothesis
                $H_{j}$,
                $z^{*}(t_{go;j'})$ is the singular region boundary
                point corresponding to the time-to-go of that
                particle, and
                $\tilde{w}^{j'}_{k}$ is the corresponding normalized
                weight at time $t_{k}$:
                \begin{equation}
                  \tilde{w}^{j'}_{k} = \frac{w^{j'}_{k}}{\sum_{j'\in H_{j}} {w}^{j'}_{k}}.
		\end{equation}}
	
            \item \textcolor{blue}{The cost of wrongly choosing
                hypothesis $H_{i}$ when $P(\meas \given H_i) >0$, when
                hypothesis $H_{j}$ is true and
                $P(\meas \given H_j) >0$, is
		\begin{equation}
			C_{ij} =  \sum_{j'\in H_{j}} \tilde{w}^{j'}_{k} \sum_{i'\in H_{i}} \tilde{w}^{i'}_{k}   
			\max (0,\abs{z_{i'j'}(t_{go;j'} - \tau)}-z^{*} (t_{go;j'} - \tau))
			\label{eq:c_ij1}
		\end{equation}
		where ${z}_{i'j'}( t_{go;j'}-\tau)$ is the ZEM of particle $j'$ that belongs to hypothesis $H_{j}$, which is propagated ahead for a predefined horizon, $\tau$.
        \textcolor{blue}{The parameter $\tau$ is the prediction duration used to evaluate the cost of the decision; in the simulation study, we set $\tau = \min(\tau_{\max}, t_{go;j'})$, where $\tau_{\max}$ is the maximal prediction duration.}
        The propagation is performed over the interval $[t_{go;j'},t_{go;j'}-\tau]$.
        During this interval, we apply the acceleration command of particle
        $i'$ that belongs to hypothesis
        $H_{i}$ using the linear acceleration command from
        Eq.~\eqref{Eq:DGL1} and assuming that the acceleration
        command is constant during the time interval.}

    \item \textcolor{blue}{The cost of wrongly choosing hypothesis
        $H_{i}$ when $P(\meas \given H_i)=0$, when hypothesis $H_{j}$
        is true and $P(\meas \given H_j)>0$, is
		\begin{equation}
			C_{ij}= \sum_{j'\in H_{j}} \tilde{w}^{j'}_{k}
			\max (0,\abs{z_{i'j'}(t_{go;j'} - \tau)} - z^{*} (t_{go;j'} - \tau)).
			\label{eq:c_ij2}
		\end{equation}}
		
            \item \textcolor{blue}{If $P(\meas \given H_j) = 0$, then
                we can set $C_{jj}$ and $C_{ij}$ arbitrarily,
                according to the optimal Bayesian decision rule of
                Eq.~\eqref{Eq:I}.}
	\end{enumerate}
		
	\item[\textcolor{blue}{Prior probabilities.}] 
	\textcolor{blue}{The prior \textcolor{blue}{probability for each hypothesis $j$, $P_j$, is} calculated based on the measurement history available at time $t_{k-1}$, denoted by $\meas_{k-1}$, prior to acquiring the new measurement at time $t_{k}$. Let SW denote the event that the target has switched its acceleration command in the time interval $[t_{k-1},t_{k}]$, and let NSW denote the complementary event. Thus, using the law of total probability yields
	\begin{equation}
		P_{j}  = \Pr(H_{j} \mid \meas_{k-1},\text{SW}) \Pr(\text{SW}) 
		+ \Pr(H_{j} \mid \meas_{k-1},\text{NSW}) \Pr(\text{NSW}).
		\label{eq:prior}
	\end{equation}
	The conditional hypothesis probabilities are obtained by propagating
	forward the particle cloud available at time $t_{k-1}$ according to
	whether a switch has occurred or not. The probabilities
	$\Pr(\text{SW})$ and $\Pr(\text{NSW})$ are obtained by using again
	the law of total probability while conditioning on the target's mode
	$r_{k-1}$ at time $t_{k-1}$
	\begin{equation}
		\Pr(\text{SW})  = \Pr(m_{k}=1 \mid m_{k-1}=2)\Pr(m_{k-1}=2)
		+ \Pr(m_{k}=2 \mid m_{k-1}=1)\Pr(m_{k-1}=1)
		\label{eq:SW}
	\end{equation}
	and
	\begin{equation}
		\Pr(\text{NSW}) = \Pr(m_{k}=1 \mid m_{k-1}=1)\Pr(m_{k-1}=1)
		+ \Pr(m_{k}=2 \mid m_{k-1}=2)\Pr(m_{k-1}=2).
		\label{eq:NSW}
	\end{equation}
	The target's modal probabilities at time $t_{k-1}$ are
        obtained from the estimator, and the transition probabilities
        populate the TPM, which is assumed to be known.}
\end{description}
\textcolor{blue}{ Algorithm~\ref{alg:estGL} presents a schematic
  description of one cycle of the usage of the Bayesian decision
  criterion for the interception problem.}
\begin{algorithm}[tbh]
  \SetAlgoNoLine \SetAlgoLongEnd \DontPrintSemicolon
  \textcolor{blue}{Obtain the weights and samples of the current and
    previous steps
    \;
    \For{$j = 1,2,3,4$}{ Calculate the likelihood of hypothesis
      $H_{j}$ via Eq.~\eqref{eq:likelihood}\;
      \eIf{$P(\meas \given H_j) > 0$}{ \For{$i = 1,2,3,4$}{
          \eIf{$i \neq j$}{ \eIf{$i\in\{2,3\}$ and
              $P(\meas \given H_i) > 0$}{ Calculate the cost $C_{ij}$
              via Eq.~\eqref{eq:c_ij1}}{ Calculate the cost $C_{ij}$
              via Eq.~\eqref{eq:c_ij2}} }{ Calculate the cost $C_{jj}$
            via Eq.~\eqref{eq:c_jj}} } Calculate the prior probability
        $P_{j}$ via Eqs.~(\ref{eq:prior}--\ref{eq:NSW}) }{ Set
        $C_{ij} = 0$ for all $i=1,2,3,4$\; Set $P_{j} = 0$} Calculate
      the expected additional risk of choosing hypothesis $H_{j}$ via
      Eq.~\eqref{Eq:I} } Choose the optimal Bayesian decision via
    Eq.~\eqref{Eq:Decision}}
	\caption{\textcolor{blue}{Usage of the Bayesian decision criterion for the interception problem}}
	\label{alg:estGL}
\end{algorithm}

\subsection{Computational Effort Reduction}
\label{sec:cond}
The computational effort for determining all the expected excess risks
of choosing a hypothesis can be significant, depending on the number
of particles used by the filter.  In this subsection, we present
conditions for determining whether the entire calculation mechanism of
the Bayesian decision criterion is needed or whether, instead, only a reduced
calculation can be employed without losing optimality.
\textcolor{blue}{We emphasize that this reduction mechanism does not involve reducing the particle population size. Instead, it avoids the computationally expensive cost integration steps when the particle cloud is entirely contained within specific regions of the game space where the optimal decision is unambiguous.}

We begin by defining the set of all particles having nonzero weights
at time $t_{k}$ as $\mathcal{J}_{k}$, that is,
\begin{equation}
	\mathcal{J}_{k} \dfn \big\{j' \mid {w}_{k}^{j'} \neq 0 \big\}.
\end{equation}
Note that the complement set of $\mathcal{J}_{k}$ has no bearing on
the computation of the costs or the likelihoods. Next, define the maximal acceleration command applied by the nonzero-weight particles in hypotheses $H_{2}$ or $H_3$ as
\begin{equation}
  \overline{u}_{i} \dfn  \arg \max_{j'\in \mathcal{J}_{k} \cap H_{i}} \left\lbrace \frac{z_{j'}(t_{go;j'})}{z^{*}(t_{go;j'})} \right\rbrace,\quad i=2,3
\end{equation}
where $|\overline{u}_{i}| < a_{M}^{\max}$.
Analogously, define the minimal acceleration command applied by the
nonzero-weight particles in hypotheses $H_{2}$ or $H_3$ as
\begin{equation}
  \underline{u}_{i} \dfn  \arg \min_{j'\in \mathcal{J}_{k} \cap H_{i}} \left\lbrace \frac{z_{j'}(t_{go;j'})}{z^{*}(t_{go;j'})} \right\rbrace,\quad i=2,3
\end{equation}
where $|\underline{u}_{i}| < a_{M}^{\max}$.  
If $\mathcal{J}_{k} \cap H_{i} = \emptyset$, then we set
$\overline{u}_{i} = \underline{u}_{i} = 0$.

We now assume the following:
\begin{enumerate}
\item \label{assmp1} If \textcolor{blue}{$P(\meas \given H_j) \neq 0$}, then
  $P_{j} \neq 0$.    
\item \label{assmp2} If
  \textcolor{blue}{$P(\meas \given H_2)+ P(\meas \given H_3) = 1$}, then
  $\overline{u}_{2} \geq \overline{u}_{3}$ and
  $\underline{u}_{2} \geq \underline{u}_{3}$.   
\item \label{assmp4} The acceleration commands of the interceptor and
  the target, $u_{M}$ and $u_{T}$, respectively, are constant in the
  interval $[t_{go},t_{go}-\tau]$.
\end{enumerate}
The first assumption states that if the likelihood indicates that a hypothesis is feasible, i.e., its likelihood is nonzero, then the prior probability must indicate that it is feasible too. 
The second assumption states that when all the particles are inside the singular region, then the particle that is closest to the upper boundary of the singular region is in $H_{2}$. Similarly, the particle that is closest to the lower boundary of the singular region must be in $H_{3}$.
Intuitively, this is because $H_{2}$ represents states tending upward and $H_{3}$ represents states tending downward, so the boundary-critical particles must lie in these respective sets.
The last assumption states that the acceleration commands are assumed to be constant for the given horizon when using the Bayesian decision criterion, as mentioned in Sec.~\ref{subsec:Bayes4int}.

The following lemma states a necessary and sufficient condition that the expected additional risk of choosing hypothesis $H_{i}$ is zero.
\begin{lemma}
At time $t_{k}$
	\begin{enumerate}
		\item $I_{1} = 0$ if and only if for $u_{M} = a_{M}^{\max}$
		\begin{equation}
			z_{1j'}(t_{go;j'}-\tau) \geq - z^{*}(t_{go;j'} - \tau)
			\label{eq:I1}
		\end{equation}
		for all $j' \in \mathcal{J}_{k}$,
		\item $I_{2} = 0$ if and only if for any $u_{M}\in [\underline{u}_{2},\overline{u}_{2}]$
		\begin{equation}
			\abs{z_{2j'}(t_{go;j'}-\tau)} \leq z^{*}(t_{go;j'} - \tau)
			\label{eq:I2}
		\end{equation}
		for all $j' \in \mathcal{J}_{k}$,
		\item $I_{3} = 0$ if and only if for any $u_{M}\in [\underline{u}_{3},\overline{u}_{3}]$
		\begin{equation}
			\abs{z_{3j'}(t_{go;j'}-\tau)} \leq z^{*}(t_{go;j'} - \tau)
			\label{eq:I3}
		\end{equation}
		for all $j' \in \mathcal{J}_{k}$,
		\item $I_{4} = 0$ if and only if for $u_{M} = - a_{M}^{\max}$
		\begin{equation}
			z_{4j'}(t_{go;j'}-\tau) \leq z^{*}(t_{go;j'} - \tau)
			\label{eq:I4}
		\end{equation}
		for all $j' \in \mathcal{J}_{k}$.
	\end{enumerate}
	\label{lem:I}
    In Eqs.~\eqref{eq:I1}--\eqref{eq:I4}, $z_{ij'}(t_{go;j'}-\tau)$ is the ZEM of particle $j'$, which is
    propagated ahead over the interval $[t_{go;j'},t_{go;j'}-\tau]$ by
    applying the acceleration command assuming that hypothesis $H_{i}$ is true.
\end{lemma}
\begin{proof}
	The proof is deferred to Appendix~\ref{proof:I}.
\end{proof}

Lemma~\ref{lem:I} provides necessary and sufficient conditions for
$I_{i} = 0$ based on future states of the particles. To develop
similar conditions based on the current-time states of the particles,
we first denote
\begin{equation}
	\eta \dfn \frac{t_{go} - \tau}{\tau_{M}} = \theta - \frac{\tau}{\tau_{M}}
\end{equation}
and then denote the following four regions over the game space.
\begin{description}
	\item[$\mathcal{A}_{1}$:] all the states satisfying
	\begin{equation}
		- \ZEM(t_{go}) \leq \partial \mathcal{Z}^{*}(t_{go}) + 2 a_{M}^{\max} \tau_{M}^{2} \left[\Upsilon(\eta)-\Upsilon (\theta)\right]
		\label{eq:I1_2}
	\end{equation}
	\item[$\mathcal{A}_{2}$:] all the states satisfying
	\begin{equation}
		\ZEM(t_{go}) \leq \partial \mathcal{Z}^{*}(t_{go})  + a_{M}^{\max} \tau_{M}^{2} \left(1 - \frac{\underline{u}_{2}}{a_{M}^{\max}}\right) \left[\Upsilon(\eta)-\Upsilon (\theta)\right]
		\label{eq:I2_up}
	\end{equation}
	and
	\begin{equation}
		- \ZEM(t_{go}) \leq \partial \mathcal{Z}^{*}(t_{go}) + a_{M}^{\max} \tau_{M}^{2} \left(1 + \frac{\overline{u}_{2}}{a_{M}^{\max}}\right) \left[\Upsilon(\eta)-\Upsilon (\theta)\right]
		\label{eq:I2_dw}
	\end{equation}
	\item[$\mathcal{A}_{3}$:] all the states satisfying
	\begin{equation}
		\ZEM(t_{go}) \leq \partial \mathcal{Z}^{*}(t_{go})  + a_{M}^{\max} \tau_{M}^{2} \left(1 - \frac{\underline{u}_{3}}{a_{M}^{\max}}\right) \left[\Upsilon(\eta)-\Upsilon (\theta)\right]
		\label{eq:I3_up}
	\end{equation}
	and
	\begin{equation}
		- \ZEM(t_{go}) \leq \partial \mathcal{Z}^{*}(t_{go}) + a_{M}^{\max} \tau_{M}^{2} \left(1 + \frac{\overline{u}_{3}}{a_{M}^{\max}}\right) \left[\Upsilon(\eta)-\Upsilon (\theta)\right]
		\label{eq:I3_dw}
	\end{equation}
	\item[$\mathcal{A}_{4}$:] all the states satisfying
	\begin{equation}
		\ZEM(t_{go}) \leq \partial \mathcal{Z}^{*}(t_{go}) + 2 a_{M}^{\max} \tau_{M}^{2} \left[\Upsilon(\eta)-\Upsilon (\theta)\right]
		\label{eq:I4_2}
	\end{equation}
\end{description}

Equipped with Lemma~\ref{lem:I} and these definitions, we now state
the following lemma.
\begin{lemma}
	There exists $i\in \{1,2,3,4\}$ such that $I_{i} = 0$ at time $t_{k}$ if and only if 
	\begin{equation}
		\sum_{j' \in \mathcal{A}_{i}} w_{k}^{j'} = 1
	\end{equation}
	\label{corol}
\end{lemma}
\begin{proof}
	The proof is deferred to Appendix~\ref{appendB}.
\end{proof}

Based on these results, we now prove the following theorem, which
provides a sufficient condition for the optimality of hypothesis
$H_{i}$.
\begin{theorem}
	If at time $t_{k}$ there exists $i\in \{1,2,3,4\}$ such that
	\begin{equation}
		\sum_{j' \in \mathcal{A}_{i}} w_{k}^{j'} = 1
	\end{equation}
	then hypothesis $H_{i}$ is an optimal Bayesian decision.
\end{theorem}
\begin{proof}
	Lemma~\ref{corol} shows the equivalence between $I_{i} = 0$ and $\sum_{j' \in \mathcal{A}_{i}} w_{k}^{j'}= 1$.
	We know that if there exists $i\in\{1,2,3,4\}$ such that $I_{i} = 0$ then hypothesis $H_{i}$ is optimal based on Eq.~\eqref{Eq:Decision}. 
\end{proof}

When the optimal Bayesian decision is a single, particular hypothesis,
there is no need to perform all the calculations as presented in
Algorithm~\ref{alg:estGL}, which translates to a significant
computational effort reduction. However, the optimal solution is not
guaranteed to be unique, and, in some cases, all four hypotheses are
optimal, which results in a decision ambiguity. We resolve this
ambiguity by harnessing it to improve the information acquired by the
interceptor, as presented in the sequel.
\subsection{Information-Enhancement Trajectory Shaping}
\label{subsec:TS}
The Bayesian decision criterion facilitates unique optimal decisions
under severe uncertainties when the associated hypothesis costs,
related to the finite cost-calculation horizon, denoted by
$\tau_{\max}$, are nonzero. However, when the expected additional
costs vanish for this horizon, no such unambiguous decision exists,
and the acceleration command becomes
arbitrary. In~\cite{mudrik_estimation-based_2019-1}, when facing such
situations, the authors arbitrarily used the classical version of the
DGL1 law.

In the present work, we propose exploiting this command ambiguity to
improve the information provided to the estimator. To that end, we
define the set of admissible control functions as
\begin{equation}
  \mathcal{U}_{M}^{k} \dfn \left\lbrace u_{M}^{k}:|u_{M}^{k}|\leq a_{M}^{\max}, \sum_{j'\in \lbrace H_{1} , H_{4} \rbrace} w_{k+1}^{j'} \leq W_{\text{Thres}} \right\rbrace, 
\label{eq:admis_set}
\end{equation}
where each member of the set, $u_{M}^{k}$, keeps most of the particles
inside the singular region, in the sense that the sum of the weights
of all particles outside the singular region at time $t_{k+1}$ does
not exceed a predefined threshold $W_{\text{Thres}}$.
\textcolor{blue}{This constraint ensures that the interceptor's state remains within the singular region. Preventing the state from exiting this region is crucial, as doing so when the target has not maneuvered can result in poor guidance performance due to the controller's inability to recover. The performance of the algorithm depends on the selection of the predefined threshold, which serves as a tuning parameter. This selection introduces a tradeoff: prioritizing stricter safety guarantees by choosing low values may result in conservative commands, whereas prioritizing performance (by choosing higher values) allows for more aggressive maneuvering but with the risk of crossing the singular-region boundaries.}

As is well known, the Cram\'er--Rao lower bound (CRLB) theorem\textcolor{blue}{~\cite{tichavsky_posterior_1998}} states
that, at time $t_k$, 
\begin{equation}
  \mathbb{E} [(\hat{\textbf{x}}_{k}-\textbf{x}_{k})(\hat{\textbf{x}}_{k}-\textbf{x}_{k})^{T}] \succeq J_{k}^{-1}
\end{equation}
where $J_{k}$, the Fisher information matrix (FIM) at time $t_{k}$, is
given by
\begin{equation}
J_{k} = \mathbb{E} [ - \nabla_{\textbf{x}_{k}}^{2} \log p (\textbf{z}_{k} \mid \textbf{x}_{k})].
\end{equation}
Based on~\cite{tulsyan_particle_2013}, a particle-based approximation
of the FIM is
\begin{equation}
J_{k} \approx - \sum_{j=1}^{N_{p}} w_{k}^{j}  \nabla_{\textbf{x}_{k}}^{2} \log p (\textbf{z}_{k} \mid \textbf{x}_{k}^{j}).
\end{equation}
Using this particle-based approximation, we can predict the resulting FIM for any sequence of acceleration commands. 
We first discretize the admissible acceleration commands to create a finite number of options. 
We then apply each of these options simultaneously as a constant acceleration command to each particle until it reaches the boundary of the singular region.  
Thereafter, each particle uses the regular bang--bang DGL1 acceleration command outside the singular region. 
When the first particle has reached the end of the engagement, we calculate the approximate FIM resulting from the applied sequence of acceleration commands.


Now, to enhance the information provided to the estimator, we seek, in
the admissible set, the interceptor acceleration command that
maximizes the expected FIM at some predefined horizon. This is
equivalent to seeking the command that minimizes the uncertainty in
the game space, represented by the CRLB (the inverse of the FIM).
Following~\cite{oshman_optimization_1999}, this command is found by
minimizing, over the admissible set, the determinant of the CRLB.
Because the dominant states in this problem are the ZEM and the
time-to-go, only these states are considered in the minimization. This
is done by first computing the entire FIM for the four-dimensional
state vector. Next, when computing the inverse of the FIM, only the
submatrix entries associated with the ZEM and the time-to-go are
considered.  \textcolor{blue}{Formally, let the inverse of the FIM for
  a predefined horizon $h$ be partitioned as
  $J_{k+h}^{-1} = \left[ \begin{smallmatrix} \mathbf{\Sigma}_{11} &
      \mathbf{\Sigma}_{12} \\ \mathbf{\Sigma}_{12}^T &
      \mathbf{\Sigma}_{22} \end{smallmatrix} \right] \in \mathbb{R}^{4
    \times 4}$, where
  $\mathbf{\Sigma}_{11} \in \mathbb{R}^{2 \times 2}$ corresponds to
  the $t_{go}$ and $Z_{\text{DGL1}}$ states. We then define
  $\bar{J}_{k+h}^{-1} \triangleq \mathbf{\Sigma}_{11}$,} and the
resulting minimization problem is, therefore,
\begin{equation}
u_{M}^{k} = \arg \min_{u_{M}^{k} \in \mathcal{U}_{M}^{k}}
\det \bar{J}_{k+h}^{-1}.
\label{eq:OHTS}
\end{equation}
%
\subsection{GST-Compliant Guidance Law}
\label{subsec:gst_gl}
To summarize the results of this section, we present herein two
versions of the modified, estimation-aware DGL1. The first version,
which includes the information-enhancement feature, is
\begin{equation}
u_{M}^{k}=
\begin{cases}
+a_{M}^{\text{max}}           & H_1 \text { is decided} \\
\sum_{j'\in H_{2}}\tilde{w}_{k}^{j'}\frac{z_{j'}( t_{go;j'} )}{z^{*}( t_{go;j'})}a_{M}^{\text{max}} & H_2 \text { is decided} \\
\sum_{j'\in H_{3}}\tilde{w}_{k}^{j'}\frac{z_{j'} ( t_{go;j'} )}{z^{*}( t_{go;j'})}a_{M}^{\text{max}} & H_3 \text { is decided} \\
-a_{M}^{\text{max}}           & H_4 \text { is decided} \\
\arg \underset{u_{M}^{k} \in \mathcal{U}_{M}^{k}}{\min}
\det \bar{J}_{k+h}^{-1}  & \text {No decision}
\end{cases}.
\label{eq:ModGL}
\end{equation}
This guidance law has five modes, the first four of which correspond
to the four hypotheses of the interception problem, where the expected
costs are nonzero. The linear acceleration command follows the
chattering-prevention logic of the DGL1 law, as shown in
Eq.~\eqref{Eq:DGL1}. The first and fourth modes correspond to cases
where the optimal hypotheses are outside the singular region. The
second and third modes correspond to cases when the chosen hypotheses are
inside the singular region.  In these cases, the linear acceleration
command is the mean of the acceleration commands of the particles
belonging to the chosen hypothesis.  The last mode addresses the case
where the expected additional risks vanish. In this case, the
information-enhancement trajectory shaping of Eq.~\eqref{eq:OHTS} is
used.  In the sequel, we use the acronym IETS, standing for
information-enhancing trajectory shaping, to denote this
estimation-aware guidance law.

The second version of this guidance law has the same first four modes
as Eq.~\eqref{eq:ModGL}; however, as its fifth mode (corresponding to
the case where the expected costs vanish for the predefined horizon
$\tau_{\max}$), it simply implements the regular version of the DGL1
law~\eqref{Eq:DGL1}.  In the sequel, we use the acronym EADGL1,
standing for estimation-aware DGL1, to denote this guidance law, which
does not implement trajectory shaping.
\section{Numerical Simulation Study}
\label{sec:sim}
We evaluate the performance of the EADGL1 and IETS guidance laws via a
numerical closed-loop simulation study. To this end, we adopt the
interception scenario presented
in~\cite{shaferman_stochastic_2016}. The performance of the new laws
is compared with that of the unmodified DGL1 guidance law.
\subsection{Simulation Scenario}
\label{subsec:param}
A ballistic missile defense (BMD) scenario is considered, which
includes a single intercepting missile and a single highly
maneuverable target. For simplicity, we assume that the target
performs a maximal-acceleration bang--bang evasion maneuver, with a
single command switch during the engagement.
 
The target is initialized in the -$Y_{I}$ direction with a flight-path
angle of $\gamma_{T}=-\pi/2$~rad, and the interceptor's flight path is
chosen such that the interceptor velocity vector points toward the
initial target location. The interceptor's and the target's maneuver
capabilities are $a_{M}^{\max}=45$~g and $a_{T}^{\max}=20$~g,
respectively. The interceptor's and the target’s first-order time
constants are $\tau_{M}=\tau_{T}=0.2$~s. The interceptor's and
the target’s speeds are $V_{M}=V_{T}=2500$~m/s, amounting to a nominal
engagement time of $3$~s.  The measurement noise is Gaussian
distributed with $\nu \sim \mathcal{N}(0,\sigma^{2})$ where
$\sigma=0.5$~mrad, and the IR sensor's sampling rate is $f=100$~Hz.

\subsection{Design Considerations}
The IMMPF uses $1000$ weighted particles to approximate the game-state posterior PDF, where each mode uses a $500$-particle population, and the filter is initialized with equal mode probabilities for both modes. 
\textcolor{blue}{As a baseline for performance comparison, we use the standard DGL1 law. This baseline utilizes the same IMMPF estimator but computes its guidance command based on a simple weighted mean of the target state estimates, ignoring the multihypothesis decision structure that is proposed in the new framework.}

The initial radar's estimate of the state vector in Eq.~\eqref{eq:RadarSV} satisfies $ \hat{x}_{R}\sim\mathcal{N}(\bar{x}_{R},P_{R}) $, where $\bar{x}_{R}$ is the true target initial state, and the associated initial covariance matrix is
\begin{equation}
	P_{R}=\text{diag} \left\{ 50^{2},\Big(\frac{1\pi}{180}\Big)^{2},
	\Big(\frac{3\pi}{180}\Big)^{2},10^{2} \right\}.
	\label{eq:P_R}
\end{equation}
To obtain the initial interceptor's estimate of the state vector in
Eq.~\eqref{eq:StateVec1}, we draw weighted samples from the radar's
estimate and use the transformation in Eqs.~\eqref{eq:trans} and
\eqref{eq:RadarInit}. We note that, following this initialization,
the distribution is no longer necessarily Gaussian due to the
nonlinearity of the transformation.  The initial estimate of
$\gamma_{T}$ is given directly from the radar's estimate.

Assuming that the target performs a single bang--bang maneuver, and
following the analysis of~\cite{shaferman_stochastic_2016}, we use a
time-varying TPM, 
with $\pi^{21}_{k}=10^{-3}$ and $\pi^{12}_{k}$ obeying the functional
form of the density of a generalized Gaussian distribution, i.e.,
\begin{equation}
\pi^{12}_{k} = 
\begin{cases}
  10^{-3}             & t_{k} \leq t_{s} \\
  c_{12} \frac{\beta}{2 \alpha \Gamma(1/ \beta)}
  \exp\left(-\frac{|t_{k}-\mu|}{\alpha}\right)^{\beta} & t_{k} > t_{s}
\end{cases}
\end{equation} 
where $\Gamma(\cdot)$ denotes the gamma function. \textcolor{blue}{Figure~\ref{p_12}
shows the behavior of the transition probability $\pi^{12}_{k}$ versus
time for the parameters $t_{s} = 1.9$, $c_{12}=0.16$, $\mu=2.5$,
$\beta=6$, and $\alpha=0.45$.}
\begin{figure}[tbh]
  \begin{center}
    \includegraphics[width=3.25in]{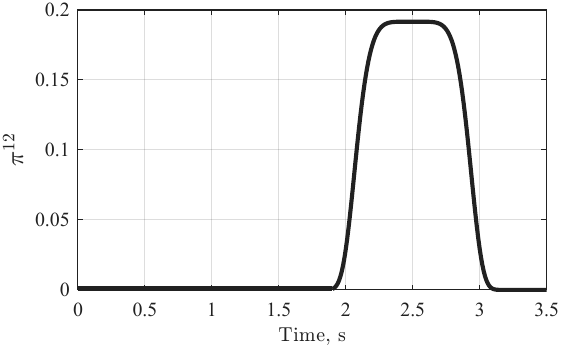}
    \caption{Time-varying transition probability from mode 1 to mode 2.}
    \label{p_12}
  \end{center}
\end{figure}


The DGL1 guidance law is used with a linear command in the singular
region and $k=0.7$. The EADGL1 law is used with a maximal horizon of
$\tau_{max} = \frac{16}{f}=0.16$~s, which is the maximal propagation
of the game's state by the Bayesian decision criterion. If a Bayesian
decision cannot be reached with this horizon, the regular version of
the DGL1 law is used.

The trajectory-shaping algorithm employed by the IETS guidance law
uses $100$ particles, which are resampled from the particle population
of the IMMPF. We set the predefined horizon to be $h = 100$, which
represents a horizon length of one second, and we do not use the
trajectory shaping in the last second of the engagement. We also set
$W_{\text{Thres}} = 0.15$.
\subsection{Single-Run Analysis}
Figure~\ref{fig:Planar_Eng_Examp} presents the trajectory of the
target and three different interceptor trajectories, corresponding to
the three guidance laws compared in this study, in a single run.  The
same initial conditions, measurement noise, and target maneuvers (a
single bang--bang maneuver at 2 s) were used in all simulations.  The
trajectories of the unmodified DGL1 and the EADGL1 laws are similar
since the modified guidance law uses the regular DGL1 whenever the
Bayesian decision criterion solution is nonunique. However, when the
IETS guidance law is used, the interceptor's course breaks away from
the trajectory generated by the unmodified DGL1 law right from the
beginning of the engagement in order to enhance the information
driving the estimator. Of course, the interceptor changes its course
back toward the target at a later point in time, so as not to lose
the target.
\begin{figure}[tbh]
	\begin{center}
		\includegraphics[width=3.25in]{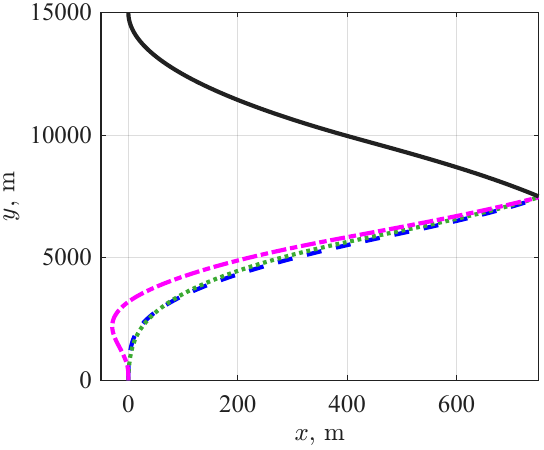}
		\caption{\textcolor{blue}{Single-run 2D trajectories: Target (solid black), unmodified DGL1 (dashed blue), EADGL1 (dotted green), and IETS (dashed--dotted magenta).}}
		\label{fig:Planar_Eng_Examp}
	\end{center}
\end{figure}

Figure~\ref{fig:Game_Space_fig} presents the trajectories from
Fig.~\ref{fig:Planar_Eng_Examp} in the game
space. Figure~\ref{fig:Game_Space_Examp} shows the state's trajectory
of the game through the entire engagement. At the beginning of the
engagement, the regular DGL1 and the EADGL1 trajectories are similar,
whereas the IETS's trajectory goes toward the singular region's
boundary. After it reaches the boundary, its trajectory moves along
it. Figure~\ref{fig:Game_Space_Examp_zoom} shows that the IETS
guidance law performs best against this maneuver when the target
performs its bang--bang maneuver due to its residence along the
trajectory. Moreover, the effect of the evasion maneuver is minimal
because the state stays inside the singular region and results in a
miss distance of 1.5~m. The miss-distance performance of the EADGL1 is
similar; however, the evasion maneuver does cause the state to cross
the singular-region boundary but with minimal effect thanks to the
Bayesian decision criterion and results in a miss distance of
3.2~m. The regular DGL1 performance is the worst compared with the
modified versions of DGL1. The evasion maneuver's effect is
significant and results in a miss distance of 7.5~m, the largest of
all three laws.
\begin{figure}[tbh]
	\begin{center}
		\begin{subfigure}{0.49\textwidth}
			\begin{center}
				\includegraphics[width=\linewidth]{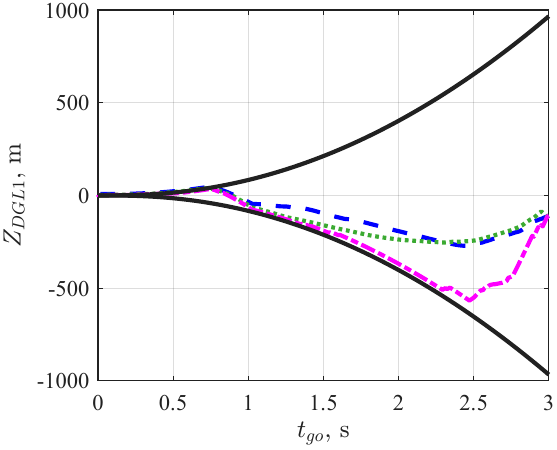}
				\caption{The entire engagement.}
				\label{fig:Game_Space_Examp}
			\end{center}
		\end{subfigure}
		~
		\begin{subfigure}{0.49\textwidth}
			\begin{center}
				\includegraphics[width=\linewidth]{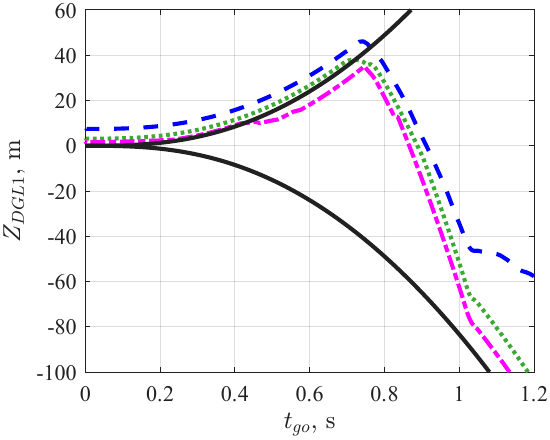}
				\caption{Zoom-in on the end of the
                                  engagement.}
				\label{fig:Game_Space_Examp_zoom}
			\end{center}
		\end{subfigure}
		\caption{\textcolor{blue}{Single-run trajectories over DGL1 game space. Comparison of the regular DGL1 (dashed blue line), EADGL1 (dotted green line), and IETS (dashed-dotted magenta line) laws.}}
		\label{fig:Game_Space_fig}
	\end{center}
\end{figure}

Figure~\ref{fig:I_per_time_Examp} compares the chosen hypothesis
throughout the engagement with the true hypothesis for the stochastic
versions of DGL1. Figure~\ref{fig:I_per_time_Examp_NO_TS} presents the
chosen hypothesis for the EADGL1; at the beginning of the engagement,
the Bayesian criterion struggles to choose the optimal hypothesis due
to severe uncertainties stemming from the initial conditions
in~Eq.~\eqref{eq:P_R}. However, in most of the next 1.6~s, there
exist more than one optimal Bayesian decision. Throughout the last 1.2~s of the engagement, the Bayesian decision criterion yields a
unique solution.  For approximately 0.2~s, hypothesis $H_{4}$ is
chosen, while hypothesis $H_{3}$ is true.  This decision represents
the conservative approach of this method, as it results in an
acceleration command ensuring that the state remains inside the
singular region.  Furthermore, $H_{4}$ is the optimal Bayesian
decision because, when the target executes its bang--bang evasion
maneuver and there is a delay in the reaction due to the severe
uncertainties, both stochastic versions of the DGL1 law outperform the
regular law, which disregards the Bayesian risk, as shown in
Fig.~\ref{fig:Game_Space_Examp_zoom}.

Figure~\ref{fig:I_per_time_Examp_TS} presents the chosen hypothesis
throughout the engagement, along with the correct hypothesis, for the
IETS guidance law. At the beginning of the engagement, the solution is
nonunique, and the trajectory-shaping technique is used.  This reduces
the duration of the ambiguous phase, as the game's state reaches the
singular-region boundary sooner, as shown in
Fig.~\ref{fig:Game_Space_fig}.  Thereafter, the state moves along the
boundary while mostly choosing $H_{4}$ to be the optimal decision,
thus exhibiting the aforementioned conservative approach.  Following
the target maneuver, the Bayesian decision criterion detects this
maneuver before the state exits the singular region, as it chooses
$H_{2}$ while it is still the true hypothesis.  Toward the end of the
engagement, when the boundaries of the singular region become severely
uncertain, the state exits the singular region.  However, as
Fig.~\ref{fig:Game_Space_Examp_zoom} shows, this still yields the
minimal miss distance, as compared to the other guidance laws.
\begin{figure}[tbh]
	\begin{center}
		\begin{subfigure}[t]{0.49\textwidth}
			\begin{center}
				\includegraphics[width=\linewidth]{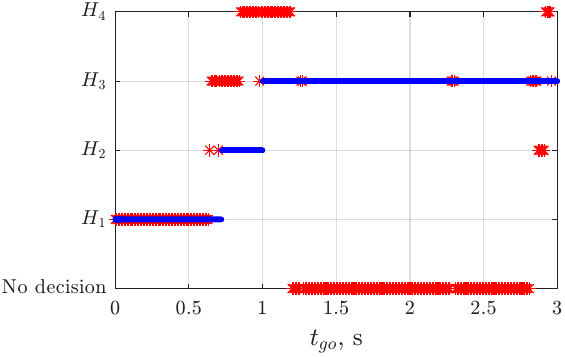}
				\caption{EADGL1 guidance law.}
				\label{fig:I_per_time_Examp_NO_TS}
			\end{center}
		\end{subfigure}
		~
		\begin{subfigure}[t]{0.49\textwidth}
			\begin{center}
				\includegraphics[width=\linewidth]{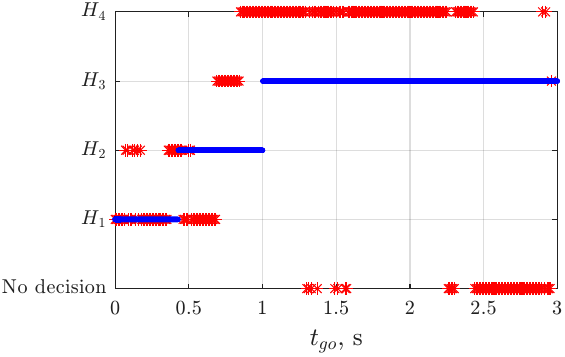}
				\caption{IETS guidance law.}
				\label{fig:I_per_time_Examp_TS}
			\end{center}
		\end{subfigure}
		\caption{The chosen hypothesis (red asterisks) and the
                  true hypothesis (blue dots) through the engagement
                  for both stochastic versions of the DGL1 law.}
		\label{fig:I_per_time_Examp}
	\end{center}
\end{figure}

Figure~\ref{fig:Bayes_demons} shows four snapshots of the particle
cloud at different times during the engagement. The actual Bayesian
decisions for these times are summarized in Table~\ref{table:1}.  In
Fig.~\ref{fig:Bayes_demons_2_5}, all the particles reside deeply inside
the singular region. Naturally, all the expected additional risks of
all the hypotheses are zero, rendering the decision nonunique.  In
Fig.~\ref{fig:Bayes_demons_1_2}, although most of the particles reside
inside the singular region, a significant number of them are under the
singular-region boundary. Thus, even though $H_{3}$ is the true
hypothesis, the optimal Bayesian decision is $H_{4}$, because it
yields the minimal expected additional risk. This case demonstrates
the aforementioned conservative nature of the algorithm.  In
Fig.~\ref{fig:Bayes_demons_0_7}, the particle cloud stretches across
all hypotheses, the estimated state resides inside the singular
region, and the true state is above the upper boundary of the singular
region.  Since the Bayesian decision criterion successfully detects
the evasion maneuver before the estimator does, it chooses $H_{1}$.
Finally, in Fig.~\ref{fig:Bayes_demons_0_07}, all the particles reside
above the upper bound of the singular region, and the Bayesian
decision is, correctly, $H_{1}$.
\begin{figure}[tbh]
	\begin{center}
		\begin{subfigure}[t]{0.49\textwidth}
			\begin{center}
				\includegraphics[width=\linewidth]{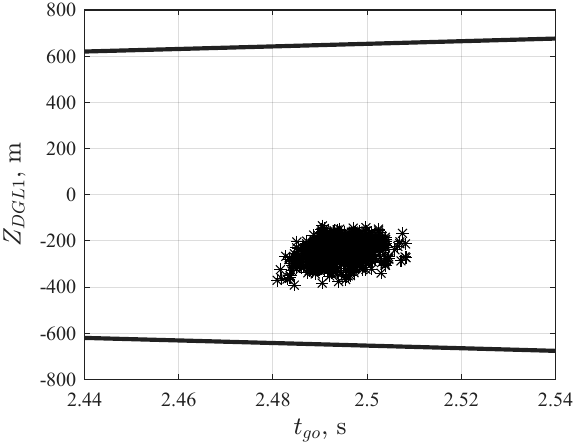}
				\caption{$t_{go} = 2.5$ s.}
				\label{fig:Bayes_demons_2_5}
			\end{center}
		\end{subfigure}
	~
		\begin{subfigure}[t]{0.49\textwidth}
		\begin{center}
			\includegraphics[width=\linewidth]{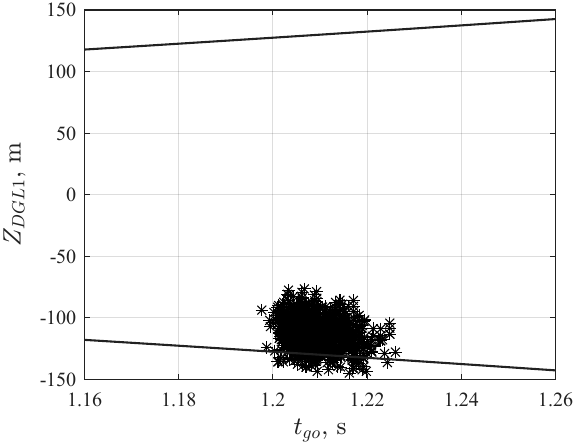}
			\caption{$t_{go} = 1.2$ s.}
			\label{fig:Bayes_demons_1_2}
		\end{center}
		\end{subfigure}
	~
	\begin{subfigure}[t]{0.49\textwidth}
		\begin{center}
			\includegraphics[width=\linewidth]{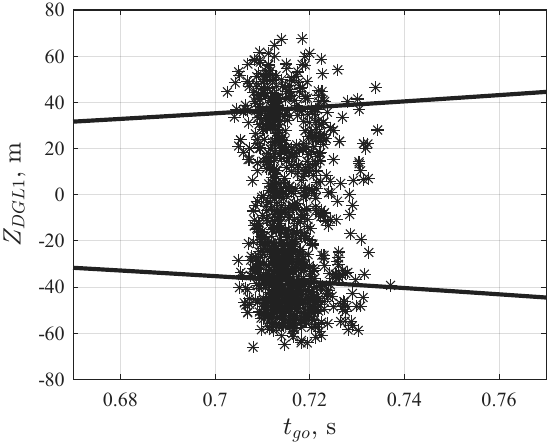}
			\caption{$t_{go} = 0.7$ s.}
			\label{fig:Bayes_demons_0_7}
		\end{center}
	\end{subfigure}
~
		\begin{subfigure}[t]{0.49\textwidth}
			\begin{center}
				\includegraphics[width=\linewidth]{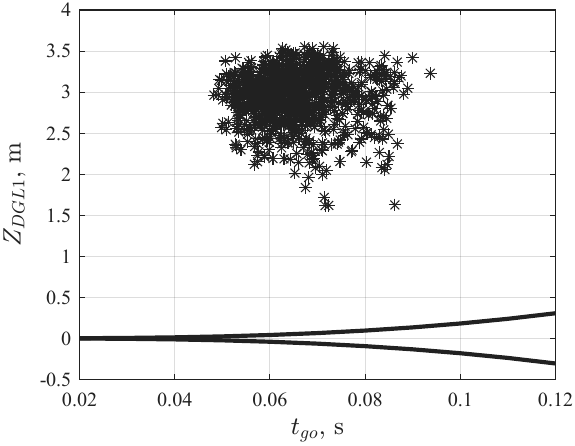}
				\caption{$t_{go} = 0.07$ s.}
				\label{fig:Bayes_demons_0_07}
			\end{center}
		\end{subfigure}
		\caption{Particle clouds over the game space
                  throughout the engagement.}
		\label{fig:Bayes_demons}
	\end{center}
\end{figure}
\begin{table}
  \caption{Bayesian decision process for the cases in
    Fig.~\ref{fig:Bayes_demons}}
\begin{center}
\begin{tabular}{c c c c c c c}
	\hline\hline
	     $t_{go}$, s      & $I_{1}$, m & $I_{2}$, m & $I_{3}$, m & $I_{4}$, m & Chosen hypothesis & True hypothesis \\ \hline
	        $2.5$         & $0$        & $0$        & $0$        & $0$        & None              & $H_{3}$         \\
	        $1.2$         & $9.96$     & $6.21$     & $0.12$     & $0$        & $H_{4}$           & $H_{3}$         \\
	        $0.7$         & $0.07$     & $0.11$     & $0.24$     & $0.66$     & $H_{1}$           & $H_{1}$         \\
	       $0.07$         & $0$        & $0.02$     & $0.03$     & $0.06$     & $H_{1}$           & $H_{1}$         \\
	[1ex] 
		\hline\hline &
\end{tabular}	
	\end{center}
	\label{table:1}
\end{table}
\subsection{Monte Carlo Analysis}
We employ an MC simulation study to compare the closed-loop performance
of the interceptor when using the regular (deterministic) and both
stochastic versions of the DGL1 law.  
\textcolor{blue}{To mitigate potential command chattering in the regular DGL1, we employ a standard chattering-reduction mechanism (see, e.g.,~\cite{shaferman_stochastic_2016}). While a complete theoretical solution to the chattering phenomenon is outside the scope of this work, our extensive MC analysis demonstrates that the guidance law remains stable and effective throughout the engagement, with no observed performance degradation due to chattering.}
In this study, the target
performs a single bang--bang evasion maneuver, timed at 0.1~s intervals
from the launch time until the end of the engagement.  In the interval
$[0, 1.5)$~s, we use 200 MC runs for each switching time.  In the
interval $[1.5, 3]$~s, we run 400 MC runs for each switching time.  This
division is made to gain more accurate statistics on the more
challenging time window of the target evasion maneuver, as identified
in~\cite{shaferman_stochastic_2016}.  Thus, a total of 9400 runs are
performed for each guidance law.

Figure~\ref{fig:MD} shows the 95th percentile of the miss distance for
each target evasion maneuver's switching time, which is the warhead
lethality radius required to guarantee a kill with a probability of
0.95.  Clearly, for all evasion maneuver switching times, both
stochastic laws outperform the regular DGL1. Moreover, the IETS law
improves the performance as the effect of the evasion maneuver becomes
more substantial, i.e., when the target maneuvers after 1.5~s.  Thus,
when the target maneuvers at 2.1~s, the IETS law reduces the required
warhead lethality radius from 14.5~m, as required for the EADGL1 law,
to 10.2~m, a $30\%$ improvement.
\begin{figure}[tbh]
	\includegraphics[width=3.25in]{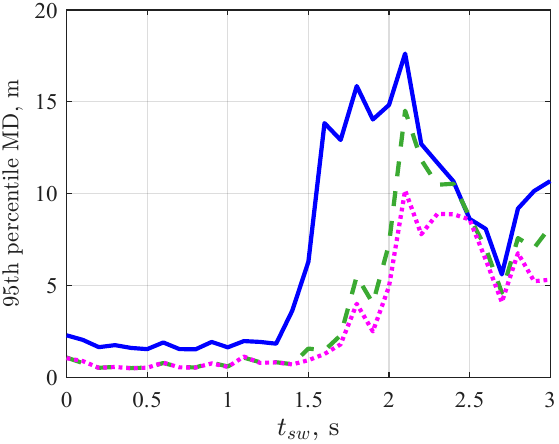}
	\centering
	\caption{\textcolor{blue}{The 95th-percentile miss distance as a function of evasion switching time for regular DGL1 (solid blue), EADGL1 (dashed green), and IETS (dotted magenta).}}	
    \label{fig:MD}
\end{figure}

Another perspective on the above findings is depicted in
Fig.~\ref{fig:CDF}, which presents the miss-distance cumulative
distribution function (CDF), as computed based on 9400 MC runs.  Both
stochastic laws outperform the regular DGL1 over the entire miss-distance
range, and the IETS law outperforms the EADGL1 law,
especially for kill probabilities larger than $0.40$.  To guarantee a
kill with a probability of 0.95, the required warhead lethality radius
would be 10.1~m for the regular DGL1, 7.09~m for the EADGL1 (a $30\%$
performance improvement), and 5.73~m for the IETS law, which is a
$43\%$ improvement compared to the regular DGL1.
\begin{figure}[tbh]
\includegraphics[width=3.25in]{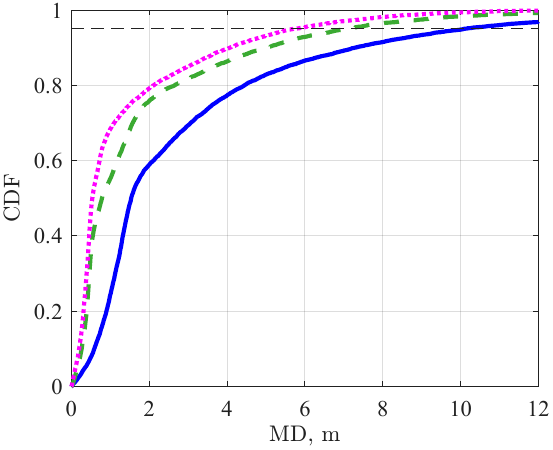}
\centering
\caption{\textcolor{blue}{Miss distance CDF in a nonlinear MC simulation. Comparison of the regular DGL1 (solid blue), EADGL1 (dashed green), and IETS (dotted magenta) laws.}}
\label{fig:CDF}
\end{figure}

Finally, a brief discussion about the computational burden associated
with the new estimation-aware algorithms is in order. Our simulations
were done on a 2020-class personal computer powered by a 6-core Intel
CPU running a non-real-time-compliant operating system. The simulation
code was written (by a nonprofessional programmer) in MATLAB 2020
using the built-in MEX mechanism and the parallel computing
toolbox. In this simulation setting, all guidance laws achieved
real-time speed.  A major contributing factor to this operation speed
was the computational-effort-reduction mechanism of
Sec.~\ref{sec:cond}.  On average, it was operational in about 82\% of
the instances where the Bayesian decision criterion was invoked, and,
when applied, it worked  11.6 times faster than the full version of
the Bayesian decision criterion.
\textcolor{blue}{Specifically, the average CPU run time for the 3~s scenario was 2.0~s for the IMMPF, 0.002~s for the regular DGL1, 0.65~s for the EADGL1 (utilizing the computational reduction), and 0.81~s for the IETS. This demonstrates that the proposed framework can be implemented in real time on standard hardware.}

\section{Conclusions}
\label{sec:con}
We have presented a novel integrated tracking and interception strategy that conforms with the generalized separation theorem. 
This strategy rests on two pillars. The first is a particle-filtering-based IMM, suitable for nonlinear, non-Gaussian, and non-Markovian hybrid systems. 
The second is a multihypothesis Bayesian decision algorithm that uses the particle cloud representing the posterior PDF to derive an optimal acceleration command that minimizes the Bayesian risk. 
Using structural properties of the Bayesian decision criterion, the computational cost of the method is reduced dramatically for most of the engagement, rendering the proposed strategy amenable to real-time implementation. 
Furthermore, when no unique Bayesian decision exists, a trajectory-shaping technique is implemented that enhances the information provided to the estimator. 
An extensive MC simulation study, involving a stochastic interception scenario, is used to demonstrate the performance advantage of the new stochastic strategy over the classical approach. 
We note that the same approach, consisting of marrying a deterministic guidance law with a Bayesian decision algorithm, can be applied to many guidance laws other than the DGL1 law used in this work. 

\begin{appendices}
\section{Proof of Lemma~\ref{lem:I}}
\label{proof:I}
We prove the lemma for each hypothesis separately.
\begin{enumerate}
	\item Recall that the expected additional risk of choosing hypothesis $H_{1}$ is
	\begin{align}
	I_{1} &= P_{2}(C_{12} - C_{22})P(\meas \given H_2)+ P_{3}(C_{13} - C_{33})P(\meas \given H_3) \notag \\
	& + P_{4}(C_{14} - C_{44})P(\meas \given H_4).
	\end{align}
	We first prove sufficiency.  The first term vanishes if
        $P(\meas \given H_2)= 0$. Otherwise, the costs $C_{12}$ and $C_{22}$ should be
        considered.  Clearly, $C_{22} = 0$ as all the particles are
        inside the singular region.  Moreover, $C_{12} = 0$, as the
        only way to cross the boundary of the singular region when
        $u_{M} = a_{M}^{\max}$ is through its lower boundary; however,
        this contradicts the condition of the lemma.  Similarly, if
        $P(\meas \given H_3) = 0$ then the second term vanishes.  Otherwise, the costs $C_{13}$
        and $C_{33}$ should be considered.  Clearly, $C_{33} = 0$ as
        all the particles are inside the singular region.  Moreover,
        $C_{13} = 0$, as the only way to cross the boundary of the
        singular region when $u_{M} = a_{M}^{\max}$ is through its
        lower boundary; however, this contradicts the condition of the
        lemma.  Considering the last term, if there exists any
        particle of nonzero weight in $H_{4}$ that sets its missile
        maneuver to be $u_{M} = a_{M}^{\max}$ when
        $u_{T} = - a_{T}^{\max}$, then the condition of the lemma is
        clearly violated, as the target uses its optimal maneuver and
        the interceptor does not. Therefore
        $P(\meas \given H_4) = 0$ which nullifies the last term.
	
	To prove necessity, assume that there exists a particle of
        nonzero weight for which the propagated ZEM violates the
        sufficiency condition, that is
    \begin{equation}
      z_{1j'}(t_{go;j'}-\tau) < -z^*(t_{go;j'}-\tau),
      \label{eq:necessity_particle}
    \end{equation}
    By Assumption~\ref{assmp2}, such a particle must belong to either
    hypothesis~$H_{4}$ or hypothesis~$H_{3}$. If it originates
    from~$H_{4}$, then $P(\meas \given H_4) \neq 0$, and, by
    Assumption~\ref{assmp1}, $P_{4}\neq 0$. In this situation, when
    the interceptor applies $u_{M}=a_{M}^{\max}$ while the target
    applies its optimal maneuver $u_{T}=-a_{T}^{\max}$, the distance
    from the singular region increases. Hence, the wrong decision
    incurs strictly larger cost, i.e., $C_{14}>C_{44}$.
    
    If, on the other hand, the particle originates from~$H_{3}$, then
    $P(\meas \given H_3) \neq 0$, and, by
    Assumption~\ref{assmp1}, $P_{3}\neq 0$. Under the same control
    actions: $u_{M}=a_{M}^{\max}$, $u_{T}=-a_{T}^{\max}$, the
    particle’s state crosses the lower boundary of the singular
    region, so that $C_{13}>C_{33}=0$. In both cases, at least one
    term in the expected additional risk becomes strictly positive,
    which yields $I_{1}\neq 0$. 
      \item Recall that the expected additional risk of choosing
        hypothesis $H_{2}$ is
	\begin{align}
          I_{2} &= P_{1}(C_{21} - C_{11})P(\meas \given H_1) + P_{3}(C_{23} - C_{33})P(\meas \given H_3) \notag \\
                & + P_{4}(C_{24} - C_{44})P(\meas \given H_4).
	\end{align}
	
	We first prove sufficiency.  If there exists a particle of
        nonzero weight in $H_{1}$ that sets its missile maneuver to be
        $u_{M} < a_{M}^{\max}$ when $u_{T} = a_{T}^{\max}$, then this
        violates the condition of the lemma, as the target uses its
        optimal maneuver and the interceptor does not. Therefore
        $P(\meas \given H_1) = 0$, nullifying the first term.  Similarly, if there exists any
        particle of nonzero weight in $H_{4}$ that sets its missile
        maneuver to be $u_{M} > - a_{M}^{\max}$ when
        $u_{T} = - a_{T}^{\max}$, then this violates the condition of
        the lemma, as the target uses its optimal maneuver and the
        interceptor does not. Therefore
        $P(\meas \given H_4) = 0$,
        which nullifies the last term.  For the second term, if
        $P(\meas \given H_3) = 0$,
        then the second term is nullified. Otherwise, we must consider
        the costs, $C_{23}$ and $C_{33}$.  Clearly, $C_{33} = 0$ as
        all the particles are inside the singular region.  Moreover,
        $C_{23} = 0$ as the option of crossing the boundary of the
        singular region contradicts the condition of the lemma.
	
	To prove necessity, we assume that there exists a particle of
        nonzero weight such that
	\begin{equation}
          \abs{z_{2j'}(t_{go;j'}-\tau)} > z^{*}(t_{go;j'} - \tau).
	\end{equation}
        This particle can originate from either $H_{1}$, $H_{3}$, or
        $H_{4}$.  In the first case,
        $P(\meas \given H_1) \neq 0$
        and by assumption~\ref{assmp1} $P_{1} \neq 0$. When this
        particle uses $u_{M} < a_{M}^{\max}$, while the target uses
        $u_{T} = a_{T}^{\max}$, the distance from the singular region
        increases, so $C_{21} > C_{11}$.  In the second case,
        $P(\meas \given H_3) \neq 0$ and by
        assumption~\ref{assmp1} $P_{3} \neq 0$. If the state of this
        particle crosses the boundary of the singular region, then
        $C_{23} > C_{33} = 0$.  In the last case,
        $P(\meas \given H_4)\neq 0$
        and by assumption~\ref{assmp1} $P_{4} \neq 0$. When this
        particle uses $u_{M} > - a_{M}^{\max}$ while the target uses
        $u_{T} = - a_{T}^{\max}$, the distance from the singular
        region increases, so that $C_{24} > C_{44}$.
	
      \item The proof of this case closely follows the proof of the
        previous case, with a few minor adjustments.
	
      \item The proof of this case closely follows the proof of the
        first case, with a few minor adjustments.
\end{enumerate}
\section{Proof of Lemma~\ref{corol}}
\label{appendB}
  Integrating the dynamic equation of the ZEM in Eq.~\eqref{eq:z_ODE}
for constant acceleration commands yields
\begin{equation}
  \ZEM(t_{go}-\tau) = \ZEM(t_{go}) - u_{M}\tau_{M}^{2}\big[\Upsilon(\theta)-\Upsilon(\eta)\big]
  + u_{T}\tau_{T}^{2}\big[\Upsilon(\theta/\epsilon)-\Upsilon(\eta/\epsilon)\big],
  \label{eq:Z_next}
\end{equation}
and the corresponding singular–region boundary satisfies
\begin{equation}
  \partial \mathcal{Z}^{*}(t_{go}-\tau) = a_{M}^{\max}\tau_{M}^{2}\Upsilon(\eta)
  - a_{T}^{\max}\tau_{T}^{2}\Upsilon(\eta/\epsilon),
  \label{eq:Z_star_next}
\end{equation}
where $\eta=\theta-\tau/\tau_{M}$.  We verify each case of
Lemma~\ref{corol} by substituting these relations into the conditions
of Lemma~\ref{lem:I}.
\begin{enumerate}
\item For $I_{1}=0$, we substitute
  Eqs.~\eqref{eq:Z_next}–\eqref{eq:Z_star_next} into Eq.~\eqref{eq:I1}
  with $u_{M}=a_{M}^{\max}$ and $u_{T}=-a_{T}^{\max}$. This shows that
  the propagated ZEM remains inside the updated singular region
  exactly when the inequality of Eq.~\eqref{eq:I1} holds, thereby
  establishing the condition.
\item For $I_{2}=0$, two subcases arise depending on the sign of the
  propagated ZEM.  If $z_{2j'}(t_{go;j'}-\tau)>0$, the worst–case
  scenario is when the target pushes the state further upward
  ($u_{T}=a_{T}^{\max}$) and the interceptor applies the least
  favorable acceleration $u_{M}=\underline{u}_{2}$. Substituting into
  Eqs.~\eqref{eq:Z_next}–\eqref{eq:Z_star_next} yields the inequality
  of Eq.~\eqref{eq:I2_up}. Conversely, if $z_{2j'}(t_{go;j'}-\tau)<0$,
  the critical case is when the target accelerates downward
  ($u_{T}=-a_{T}^{\max}$) and the interceptor uses
  $u_{M}=\overline{u}_{2}$, yielding the inequality of
  Eq.~\eqref{eq:I2_dw}. Thus both directions are covered, and the
  condition for $I_{2}=0$ follows.
\item The case $I_{3}=0$ is analogous to the previous one but applies
  to particles associated with hypothesis $H_{3}$. The same
  reasoning—splitting according to the sign of the propagated ZEM and
  considering the corresponding worst–case controls—leads to the
  conditions in Eqs.~\eqref{eq:I3_up}–\eqref{eq:I3_dw}.
\item Finally, for $I_{4}=0$, we substitute
  Eqs.~\eqref{eq:Z_next}–\eqref{eq:Z_star_next} into Eq.~\eqref{eq:I4}
  with $u_{M}=-a_{M}^{\max}$ (and the target taking its adversarial
  control). This yields precisely the inequality required in
  Eq.~\eqref{eq:I4}, completing the proof.
\end{enumerate}
In all four cases, the integrated ZEM dynamics together with the
updated singular–region boundaries confirm that the conditions of
Lemma~\ref{corol} are satisfied under the respective worst–case
control selections.

\end{appendices}

\bibliography{Bib1}

\end{document}